\documentclass[journal]{IEEEtran}
\usepackage{graphicx}
\usepackage{float}
\usepackage{subcaption}
\usepackage{epstopdf}
\usepackage{blindtext}
\usepackage{xcolor}
\usepackage[hidelinks]{hyperref}
\usepackage{enumitem}
\usepackage{multirow}
\usepackage{amsthm, amssymb, amsmath, nccmath}
\usepackage{bm}
\usepackage{mathtools}
\usepackage{cite}
\usepackage{graphics}
\usepackage[ruled,vlined]{algorithm2e}
\usepackage{array}
\newcolumntype{P}[1]{>{\arraybackslash}p{#1}}

\DeclarePairedDelimiter\floor{\lfloor}{\rfloor}

\DeclareCaptionFont{red}{\color{red}}

\newtheorem{definition}{Definition}

\newtheorem{proposition}{Proposition}

\DeclareUnicodeCharacter{2212}{-}

\begin{document}
\title{Achieving QoS for Real-Time
Bursty Applications over Passive Optical Networks}
	\author{Dibbendu Roy, Aravinda S. Rao, Tansu Alpcan, Goutam Das and Marimuthu Palaniswami}
	
	\author{
		Dibbendu\;Roy,
		Aravinda.\;S.\;Rao,\;\IEEEmembership{Member,\;IEEE,}
		Tansu\;Alpcan,\;\IEEEmembership{Senior\;Member,\;IEEE,}
		Goutam Das, 
		and\;Marimuthu\;Palaniswami,\;\IEEEmembership{Fellow,\;IEEE,}
		\thanks{}
		\thanks{D. Roy, A. S. Rao, T. Alpcan and M. Palaniswami are with the Department
			of Electrical and Electronic Engineering, The University of Melbourne, Parkville, VIC 3010, Australia (e-mail: dibbendur@student.unimelbedu.au, aravinda.rao@unimelb.edu.au, tansu.alpcan@unimelb.edu.au, palani@unimelb.edu.au).}
		\thanks{G. Das is with the G.S Sanyal School of Telecommunication, Indian Institute of Technology (IIT) Kharagpur, India (email: gdas@gssst.iitkgp.ac.in).}
		\thanks{This research was supported in part by the Australian Government through the Australian Research Council's Discovery Projects funding scheme (project DP190102828).}
	}
	\maketitle
	\begin{abstract}
	Emerging real-time applications such as those classified under ultra-reliable low latency (uRLLC) generate bursty traffic and have strict Quality of Service (QoS) requirements. Passive Optical Network (PON) is a popular access network technology, which is envisioned to handle such applications at the access segment of the network. However, the existing standards cannot handle strict QoS constraints. The available solutions rely on instantaneous heuristic decisions and maintain QoS constraints (mostly bandwidth) in an average sense. Existing works with optimal strategies are computationally complex and are not suitable for uRLLC applications.
    This paper presents a novel computationally-efficient, far-sighted bandwidth allocation policy design for facilitating bursty traffic in a PON framework while satisfying strict QoS (age of information/delay and bandwidth) requirements of modern applications. 
	To this purpose, first we design a delay-tracking mechanism which allows us to model the resource allocation problem from a control-theoretic viewpoint as a Model Predictive Control (MPC). MPC helps in taking far-sighted decisions  regarding  resource  allocations  and  captures the time-varying dynamics of the network. We provide computationally efficient polynomial-time solutions and show its implementation in the PON framework. Compared to existing approaches, MPC reduces delay violations by approximately 15$\%$ for a delay-constrained application of 1ms target. Our approach is also robust to varying traffic arrivals.

	\end{abstract}
	
	\IEEEpeerreviewmaketitle
	\section{Introduction}
	\label{intro}
Modern networking applications such as Augmented Reality/Virtual Reality (AR/VR), haptic communication and autonomous vehicles \cite{giordani2020toward,chowdhury20196g,fitzek2020computing,philip2018distributed} are typical use cases for
real-time IoT. International Telecommunication Union broadly categorize mission critical applications under ultra-reliable low latency (uRLLC) which require strict delay constraints to be satisfied with high reliability of 99.99\%. These applications demand real-time service with stringent delay constraints in the order of sub milliseconds along with high bandwidth requirements. This is also seen in industrial applications, such as the Smart Grid (3-20 ms) and Smart Factory (0.25-10 ms) settings \cite{schulz2017latency,da2014internet,sisinni2018industrial}. It has already been identified that the traffic pattern \cite{apicharttrisorn2020characterization,mondal2020enabling} for these applications do not resemble those of simple voice/video and data services, currently catered by existing networking standards \cite{itu2008g,ieeestd}. These applications are event driven and exhibit bursty traffic situations for which existing resource allocation standards are inefficient. This paper aims at designing a novel mechanism to obtain optimal resource allocation policies for bursty, delay-constrained and bandwidth-demanding applications over access networks.

Among the existing access networking technologies, Passive Optical Networks (PONs) have been identified as a promising solution for the high speed access network owing to its high bandwidth and low energy consumption \cite{rimal2017cloudlet,rimal2018experimental,helmy2018feasibility,helmy2018toward,comparewith,neaime2018resource}.
A PON (see Fig. \ref{poniot}) typically consists of ONUs (connected to users). The fibers from ONUs are aggregated at a passive power splitter called remote node. A feeder fiber with high operating bandwidth connects the remote node to a central office called Optical Line Terminal (OLT). To facilitate 
emerging uRLLC applications in a PON based framework, a fog/edge node can be plugged at the remote node (to reduce round-trip delay compared to plugging at OLT) in an overlay fashion \cite{rimal2017cloudlet,roy2020cost,comparewith} as shown in Fig. \ref{poniot}. 
In such a framework, similar to the OLT, fog node has the important task of allocating bandwidth to the ONUs so that the strict Quality of Service (QoS) requirements of uRLLC applications can be met. Important QoS parameters may include delay, age of information (AoI)\cite{zhou2019minimizing}, bandwidth/throughput and jitter. Evidently, the allocation decisions at a given instant affects future decisions. To design optimal policies satisfying strict QoS constraints, the following requirements must be taken into consideration:
\begin{itemize}
    \item The designed policy must be able to handle the bursty traffic nature of modern networking applications
    \item These policies must be far-sighted i.e., designed with a purview of the future and overall system dynamics.
    \item For real-time implementations, the algorithm to find such policies must be computationally tractable.
\end{itemize}

	
	\begin{figure}
		\centering
		\includegraphics[scale = 2]{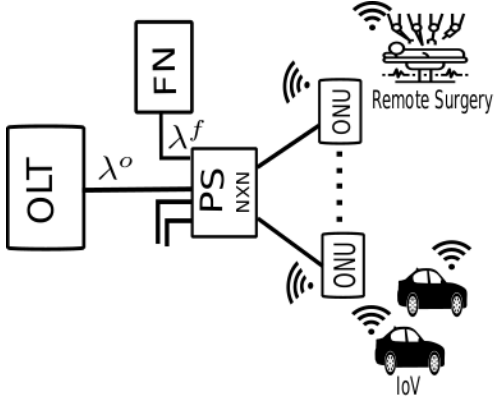}	
		\caption{Optical Distribution Network to facilitate applications like remote surgery and internet of vehicles (IoVs). Fog Node (FN) and OLT use different wavelengths for communicating with ONUs. FN is plugged to any one of many available ports of an $N\times N$ passive power splitter (PS).}
		\label{poniot}
	\end{figure}
	
	In PON, the existing standards \cite{itu2008g,ieeestd} do not have suitable mechanisms to deal with the QoS constraints under bursty traffic conditions. The existing proposals which are employed by the OLT for managing QoS (bandwidth) requirements \cite{assi2003dynamic,itu2008g,ieeestd}, may be applied at a fog node. These proposals, similar to the available standards, are based on broad traffic classification and employ heuristics based on the currently available information and only consider delays in average sense. Thus, the existing policies cannot handle bursty traffic with strict QoS requirements and are short-sighted.
	
	In past decade, several methods have been proposed to maximize throughput in a network. Under simplifying assumptions, they come up with policies which are provably optimal \cite{lin2006tutorial,ying2010combining,singh2018throughput}. However such policies perform poorly when strict delay constraints are to be maintained \cite{singh2018throughput}. As pointed out in \cite{singh2018throughput}, in order to maintain strict delay constraints and design far-sighted policies, one should consider solving Markov Decision Process (MDP) based models with large state spaces, which are computationally intractable.
	It is evident that, the resource allocation problem under QoS constraints leads to complex formulations and computationally intractable policies. Due to the high complexity of computing optimal allocation policies, most resource allocation rules are based on heuristic decisions inherently depend on currently available information.

	
In summary, \textit{to the best of our knowledge, the existing standards cannot handle strict QoS constraints for applications generating bursty traffic. The available proposals are  short-sighted heuristics and maintain QoS constraints (mostly bandwidth) in average sense. The proposals with optimal strategies are computationally complex and are not suitable for uRLLC applications}. Thus, it is important to model and design allocation policies under strict delay constraints which are far-sighted, easy to compute and are provably optimal. A detailed literature review is provided in Section \ref{relpon} identifying the drawbacks of existing proposals.

To meet the stringent delay constraints of real-time Internet and uRLLC applications in a polling based wire-line system (like PON), we develop a generic delay-tracking mechanism by employing virtual queues. This allows us to pose the resource allocation problem from a control theoretic perspective as a Model Predictive Control (MPC) problem. In contrast to MDP formulations, the MPC problem does not have stochastic state-action relations, leading to reduced state-action space. MPC helps in taking far-sighted decisions regarding  resource  allocation  rather than short-sighted (decisions based only on currently available information) and  captures the time-varying dynamics of the network. The MPC problem turns out to be an Integer Linear Program (ILP). We prove that the ILP for the short-sighted scenario can be solved by an equivalent max-flow problem \cite{papadimitriou1998combinatorial} which has several computationally efficient implementations. For the generic MPC problem, we show that linear program (LP) relaxation of the ILP yields optimal solutions. This is a significant reduction in complexity which is desired for providing delay-constrained services as the problem can be solved in polynomial time. We employ the developed technique in a PON framework as shown in Fig.\ref{poniot}. Although we apply the methodology in a PON framework, the solution strategies would be applicable for any wire-line polling based network system.
	\subsection{Contributions of the paper}
	\label{contri}
	The main contributions of the paper can be summarized as:
	\begin{itemize}
	\item We present a delay-tracking mechanism with help of virtual queues enabling us to formulate the problem of resource allocation, subject to delay and bandwidth constraints as an MPC.
	\item We show that the short-sighted scenario of the problem can be solved by an equivalent max-flow representation while the generic far-sighted problem can be solved by an equivalent LP. Thus, efficient implementation of the MPC is guaranteed with low complexity.
	\item We show how our developed model can be implemented over an overlay PON.
	\item Compared to existing proposals, our results show significant improvement in reducing in percentage of delay violations by approximately 15\% and increasing the throughput at the cost of average delay performance. Our approach is robust to varying traffic arrivals.
	\end{itemize}
	
	\subsection{Relevant Works and Literature Gaps}
	\label{relpon}
	We briefly mention the relevant literature in context of resource allocation. As discussed in Section \ref{intro}, we classify the available literature in three categories - available standards for PON 1)available standards for PON  2) specific to PON and emerging applications 3) general networking. 
	The  Giga-bit PON (GPON) \cite{itu2008g}  and Ethernet PON (EPON) \cite{ieeestd} are the two available PON standards. Although these differ majorly in terms of specifications (like frame structures and data rates), the philosophy of dynamic bandwidth allocation in both are similar. In both GPON and EPON, OLT allocates bandwidth using a polling mechanism where the allocation decisions are sent in the downstream. The ONUs send the status of their buffers and QoS requirements along with their data transmissions in the upstream. In GPON, ONUs are equipped with virtual buffers called transmission containers (T-CONTs) for different kinds of services. However, for strict delay requirements (real-time), only constant bit-rate (CBR) traffic can be accommodated. Same is the case for EPON, where traffic is classified into three broad categories: Expedited Forwarding (EF), Assured Forwarding (AF) and Best Effort (BE). Each ONU has three virtual buffers classifying its traffic in either of the three. EF deals with CBR (similar to GPON) while AF deals with assured bandwidth guarantees without strict delay guarantees.
	
	The existing proposals in PON for real-time Internet and Emerging applications \cite{comparewith,ou2017resource} mostly deal with network architecture and mere feasible allocation propositions based on available standards. The authors of \cite{comparewith,helmy2018toward} propose a decentralized protocol which employs an out of band (OOB) wavelength for offloading fog data. In \cite{comparewith}, each ONU uses the out of band wavelength for offloading its fog node upstream only during its turn for OLT upstream. Hence, the amount of bandwidth allocated for fog upstream is restricted by the duration of OLT upstream transmission. The major drawback of such a scheme is that the fog node upstream scheduling does not depend on the fog node-traffic, instead it depends on the OLT traffic. A more justified treatment was provided by the authors of \cite{ou2017resource}, where they employ slicing of the available bandwidth and reserve a portion of the bandwidth for fog services. The amount to be reserved could be decided by means of prediction or learning mechanisms. Also, the fog services (mentioned as migratory services by the authors) are prioritized over the normal best effort services. In \cite{roy2020cost}, we proposed a bandwidth allocation protocol for the fog node which schedules the fog services when the ONUs are free of its OLT upstream. 
	
	In networking literature, several works discuss the problem of designing routing, scheduling and allocation policies while maximizing throughput. The authors of \cite{singh2018throughput} provide a comprehensive idea of the available works in the area and suggest that under unconstrained scenarios certain policies have been proved to be throughput optimal. The authors have shown a generic networking scenario with routing and power allocation policies formulated as an Markov Decision Process (MDP). The requirement for scheduling was largely avoided by assuming no contention in the links and no link capacity constraints. Under these simplifying assumptions they reduce the MDP to a LP. In general, on invoking other constraints, the problem remains computationally intractable and not suitable for emerging applications.



A fog service provider operates a network with a given capacity and has to serve multiple services with different service level agreements (SLAs). Athough the fog node can employ existing dynamic bandwidth allocation (DBA) protocols which aim at maintaining QoS at OLT, they suffer from the following shortcomings: 
	
    \begin{itemize}
		\item 
		
		Existing standards based on broad classification of the traffic will not be suitable for bursty traffic with strict delay constraints. A generic model for satisfying different delay and bandwidth requirements is currently absent.
		\item 
        All existing proposals take short-sighted heuristic decisions and have considered delay constraints in an average sense. Evidently, such proposals fail to guarantee strict constraints and are non-optimal.
		\item To satisfy strict delay constraints, a complex MDP based model is usually developed which are not only computationally intractable but also unsuitable for uRLLC applications..
	\end{itemize}

	The rest of the paper is organized as follows. We describe the considered system model in Section \ref{sysmod} which contains details of our proposed delay-tracking mechanism followed by a brief description of control theoretic tool - MPC. Section \ref{probdef} poses the resource allocation problem as an MPC followed by its solution in Section \ref{solmethod}. The details of how the obtained solution can be employed in a PON framework is discussed in Section \ref{impl}. Finally, we present our simulation results in Section \ref{sim} and conclude the paper with insightful remarks and future directions in Section \ref{conc}.
	\section{PON System Model}
	\label{sysmod}
	We consider a PON with operating bandwidth of B bits per second. IoT services have different QoS requirements which are broadly classified as certain classes of service (CoS) \cite{giordani2020toward,chowdhury20196g}. Each class of service corresponds to a different delay and bandwidth requirement. The bandwidth and delay requirement for each class $c \in C$ ($C$ denotes the set of classes) is denoted by $(b_c,d_c)$ where $b_c$ is in bps and $d_c$ is in seconds (see Table \ref{tab}). The bandwidth requirement is to be maintained for a class of service over some time. We assume that the delays due to wireless access and processing of tasks at fog node can be estimated and accounted for while considering the delay constraint. An important QoS parameter in case of IoT devices is the age of information (AoI) \cite{zhou2019minimizing}. It is defined as the time elapsed since the most recent received status from an IoT device. For a polling based scheme, AoI is equivalent to the duration of a polling cycle.
	\begin{table}[!ht]
			\caption{Description of Notations}
			\label{tab}
			\begin{tabular}{|P{0.2\linewidth} | P{0.7\linewidth}|}
				\hline
				\multicolumn{1}{|c|}{\textbf{Notation}} &  \multicolumn{1}{|c|}{\textbf{Description}}  \\  
				\hline
				$N$ & Set of ONUs\\
				\hline
				$C$ & Set of classes\\
				\hline
				$B$ & Bandwidth of PON   \\
				\hline
				$b_c$ & Bandwidth requirement for class c in bps\\
				\hline
				$d_c$ & Delay requirement for class c in secs\\
				\hline
				$A^c(t)$ &Number of arrivals of class $c$ in a slot $t$ (Fig. \ref{vq})\\
				\hline
				$T_s$ &Duration of a time slot\\
				\hline
				$P_s$ & Packet size in bits\\
				\hline
				$\Lambda = \floor{\frac{BT_s}{P_s}}$ & Maximum number of packets that can be cleared in $T_s$\\
				\hline
				$\Lambda^c $ & Maximum number of packets that can be cleared for to maintain bandwidth constraint for class $c$ over horizon $H$\\
				\hline
				$Q_i^c(t)$ & $i^{th}$ Queue for tracking delay of class $c$ and time slot $t$ \\
				\hline
				$x_i^c$ & Allocation corresponding to $Q_i^c(t)$ obtained by solving \eqref{mpcprob2}\\
				\hline
				$K^c$ & Number of virtual queues for tracking delay constraint $d_c$ \\
				\hline
				H & Considered horizon for MPC\\
				\hline
				$rtt_i^f$ & Round trip from fog to $ONU_i$\\
				\hline				
			\end{tabular}
		
	\end{table}

	We consider plugging fog node at the remote node \cite{roy2020cost}. It is attached to the available free ports of the passive power splitter as shown in Fig. \ref{poniot}.The fog node communicates with the ONUs (equipped with separate transceivers for fog and OLT communication) at a wavelength different from that of the OLTs, thereby avoiding collisions among them. The bandwidth allocation for each ONU is decided by the fog node. The fog node, similar to the OLT, employs a standard polling based scheme. The details of how the fog node implements the allocation protocol will be discussed in Section \ref{impl}.  By means of the polling mechanism, the ONUs inform their demands for each class to the fog node. Once this information is available at the fog node, it decides on how these demands are to be satisfied such that the overall throughput is maximized without violating the QoS constraints.
	
	
	Since we have strict delay constraints, it is evident that delay for packets need to be tracked. However, tracking delay for each packet requires large storage requirements which might not be available at facilities such as a fog node. By considering a slotted time system, the delay experienced by packets in a slot can be summarized. We present a detailed description of our delay-tracking mechanism which would aid in formulating the allocation problem.
	\begin{figure}
		\includegraphics[width=\linewidth]{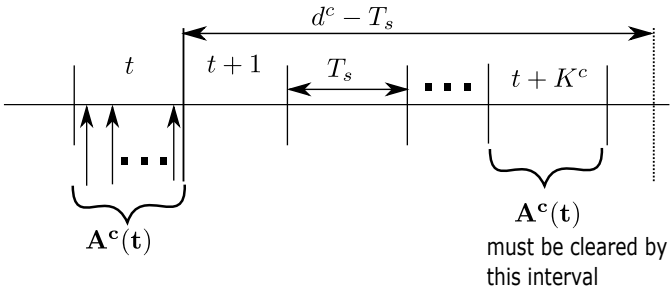}
		\caption{Finding $K^c$ as the maximum number of slots by which arrivals for a class must be cleared.}
		\label{kc}
	\end{figure}
	
	\subsection{delay-tracking Model using Virtual Queues}
	\label{delaytrack}
	\begin{figure}
		\includegraphics[scale=0.7]{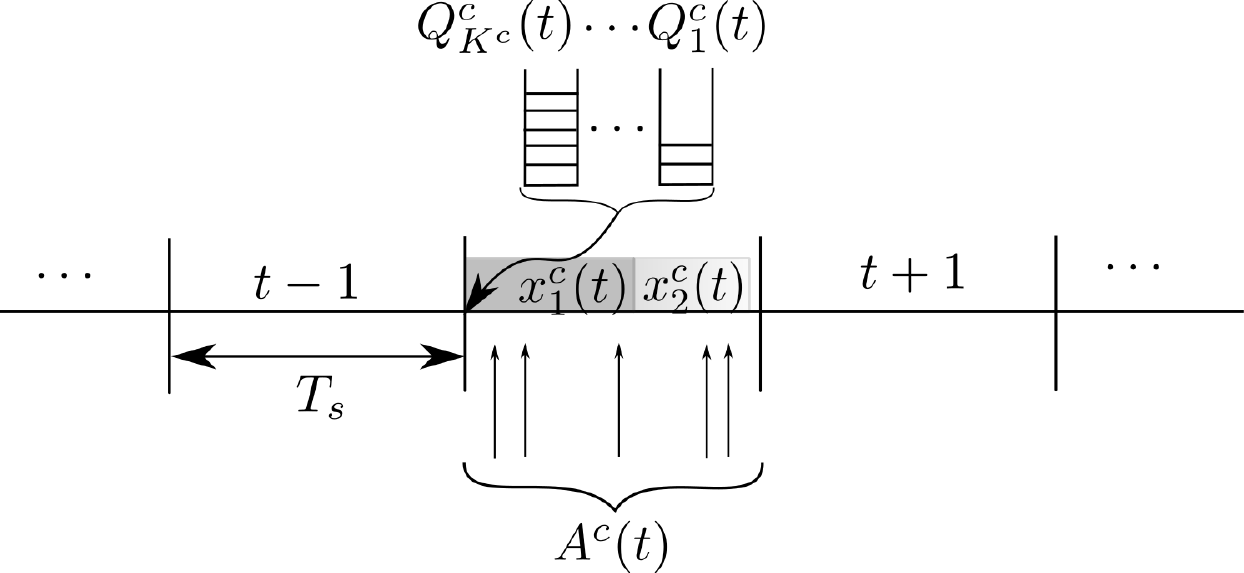}
		\caption{Mechanism for tracking delay using virtual queues}
		\label{vq}
	\end{figure}
	To develop the model, we consider a generalized class $c$ with requirements $(d_c,b_c)$. The model can be easily extended for multiple classes by varying $c$ and introducing consequent variables for the same. We consider that time is slotted with slot duration of $T_s$ (see Fig. \ref{kc}) and employ a virtual queuing scheme to track the delay requirements. We observe the system at a slot boundary. The slots are indexed in time as $t-1,t,t+1, \dots$ and so on. The number of arrivals in each slot is denoted by $A^c(t)$ as shown in Fig. \ref{kc}. We consider that the aggregated queue status at the start of a slot $t$ is given by $Q^c(t)$ and  as shown in Fig. \ref{vq}, the number of slots until which the fresh arrivals in a slot must be cleared in order to satisfy the delay requirements is (see Fig. \ref{kc}):
	\begin{equation}
	K^c = \floor{\frac{d_c-T_s}{T_s}}
	\end{equation}
	
	At any given time slot $t$, the packets can be reclassified based on the slots within which it must be cleared without violating the delay constraint. Thus, we have $K^c$ virtual queues $Q_1^c,\dots, Q_{K^c}^c$ with $Q_1$ holding the packets which must be cleared in the immediate slot, whereas $Q_{K^c}^{c}$ contains the fresh arrivals which can be cleared by $K^c$ slots. At the start of a slot say $t$, we denote the state of the $i^{th}$ queue or the number of packets that can be cleared within $i$ slots from the slot boundary by $Q_i^c(t)$.
	
	If $x^c_i$ denotes the number of packets that is cleared in slot $t$, we have the following update equations
	\begin{subequations}
		\label{stateeq}
		\begin{align}
		Q_{i-1}^c(t+1) &= Q_i^c(t)-x_i^c(t) \;\; \forall \;\; 1 < i \leq K^c \label{vq1}\\  
		Q_{K^c}^c(t+1) &= A^c(t) \label{vqi}
		\end{align}
	\end{subequations}
	where \eqref{vq1} indicates that the packets that remain uncleared at the end of $t+1^{th}$ slot from the $i^{th}$ queue have encountered a delay of $T_s$ and hence should be cleared within $i-1$ slots. \eqref{vqi} indicates that fresh arrivals in a slot can be cleared within $K^c$ slots. 
	Although in the model we consider the arrivals $A^c(t)$ to be known, the same may have to be estimated in reality. The effect of estimation errors are shown to be negligible by means of a sensitivity analysis in our results. 
	We present the basics of MPC leading to the formulated problem.
	\subsection{Model Predictive Control (MPC)}
	\label{mpc}
	MPC is used for controlling dynamic systems while satisfying a set of constraints. As discussed in Section \ref{intro}, an MPC looks ahead in time and obtains efficient far-sighted policies.
	
	An MPC problem is of the form:
	
	\begin{subequations}
		\begin{gather}
		\min J = \sum_{t=0}^{N-1} l(\mathbf{s}(t), \mathbf{a}(t)) \\
		\textrm{s.t.}  \quad \mathbf{s}(t+1) = \mathbf{A} \: \mathbf{a}(t) + \mathbf{B} \: \mathbf{s}(t) \;\; \forall t \label{mpstate} \\
		\mathbf{a(0)} = \mathbf{z}, \: \mathbf{a}(t) \in \mathbb{A} \textrm{ and } \mathbf{s}(t) \in \mathbb{S}
		\end{gather}
	\end{subequations}
	where $\mathbf{s}(t)$ are \textit{state variables} which are controlled by \textit{action variables}  \textit{$\mathbf{a}(t)$}. The function $l(\cdot)$ is a loss function which is to be minimized over the action variables $\mathbf{a}(t)$. The time index $t$ are the discrete time intervals over which the system is observed. $\mathbf{a(0)} = \mathbf{z}$ represents the initial action at the start of observations. $N$ is the \textit{horizon} or future time steps until which the system is being observed. Ideally the horizon should be infinite, however, for computational tractability, a finite horizon is usually implemented. The matrices $\mathbf{A}$ and $\mathbf{B}$ describe the state equation which essentially shows how the future state depends on the current states and actions. $\mathbb{A}$ denotes the set of feasible actions and $\mathbb{S}$ denotes the set of feasible states. In a typical MPC problem, all the associated sets are assumed to be convex and the loss function to be a convex function for ease of implementation. The optimization is solved at each time step to obtain the actions which generate new states by following the state equation.

	\section{Problem Formulation}
	
	\label{probdef}
	As discussed in Section \ref{intro}, optimal far-sighted policy design for delay-constrained services require MDP based formulations involving stochastic state-action transition matrices which become computationally intractable for large state and action spaces. The aforementioned delay-tracking mechanism shows that the queue states and corresponding allocations can be easily cast into linear deterministic state equations \ref{vqi} and \ref{vq1}. This allows us to cast the resource allocation problem as a MPC discussed in \ref{mpc}. We define the states, actions/policies, time step, horizon and all specifics of an MPC problem to fit our system model.
	
	We use the following notation scheme: the subscript denotes the index of the queue under consideration while the superscript denotes the index of the corresponding class. In this section, we still proceed with single class case as the multi-class can be simply presented as an extension of the same. However, to emphasize the notion of classes, the superscript $c$ has not been omitted.
\subsection{Formulating Resource Allocation as MPC}
	\textbf{State Variables:} We consider the queue status as the state variable $Q_i^c(t)$. The state equations for the same have already been presented in \eqref{stateeq}.
	
	\textbf{Actions Variables:} The allocations are the action variables $x_i^c(t)$. It is important to note that the allocations are in terms of packets or bits and hence are non-negative integers $x_i^c(t) \in \mathbb{Z^+}$. This restriction prohibits our problem to be a conventional convex optimization problem. 
	
	\textbf{Time step:} We represent each time step by $T_s$. Since in PON, the polling scheme requires at least a round trip time, we consider $T_s$ to be at least equal to the maximum round trip time from the fog node.
	\begin{equation}
	T_s \geq \max\{rtt^f_i  \;\; \forall i\}\label{minstep}
	\end{equation}
	
	\textbf{Horizon:} The horizon is the number of future time steps considered for performing the optimization. We denote the horizon by $H$. The case for $H = 0$ is the short-sighted version of the problem. Increasing the horizon implies looking farther ahead in the future at the cost of computational complexity.

	\textbf{Objective Function:} In order to optimize throughput in PON, the sum throughput is maximized. The sum throughput is given by:
	\begin{equation}
	\max_{x_{i}^c(t) }  \;\; \sum_{t=0}^{H}\sum_{i = 1}^{K^c} x_i^c(t)
	\end{equation}
	
	It is important to note that our delay-tracking mechanism allows us to guarantee delay constraints. It is evident that the first queue ($Q_1^c(0)$) for the first time slot (indexed by $0$), must be cleared unless the number of packets in the queue exceeds the maximum number of packets that can be cleared within a slot. If $\Lambda$ (refer Table \ref{tab} denotes the maximum number of packets that can be cleared within $T_s$, then 
	\begin{equation}
	x_1^c(0) = \min\{Q_1^c(0),\Lambda\} \;\forall \;t
	\end{equation}
	This ensures that the packets which are most critical and are about to violate delay constraints are cleared, thereby avoiding delay violation. Although we may fix the allocation $x_1^c(0)$ from the queue $Q_1^c(0)$ , the same cannot be done for the future allocations. This is because the future allocations depend on future queue states that evolve using the state equations. Thus, the sum throughput is computed for all $i$ and $c$ except for $x_1^c(0)$. Effectively, the objective is:
	\begin{equation}
	\max_{x_{i}^c(t) }  \;\; \sum_{t=0}^{H}\sum_{i = 1}^{K^c} x_i^c(t) - x_1^c(0)
	\label{mpcobj}
	\end{equation}
	
	This objective function indicates that the MPC would try and maximize the throughput while minimizing the allocation for $x_1^c(0)$. This implies that other queues would be emptied as much as possible so that the first queue is left with minimum number of packets.
	
	\textbf{Constraints:} Since in MPC, we look ahead in time, the state evolving equations described in \eqref{vqi} come into play. In addition, the maximum number of packets that can be cleared in a slot is bounded by $\Lambda$ for any slot $t$. Thus, $\forall t$, we have:
	\begin{align}
	\sum_{i = 1}^{K^c} x_i^c(t)\leq \Lambda  \;\;\forall \;t
	\label{bwcon1}
	\end{align}
	
	For a given class $c$, the corresponding bandwidth constraint $b_c$ should not be violated over time. Since, we look ahead in time, the following constraint should be implemented.
	\begin{equation}
	\sum_{t=0}^{H}\sum_{i = 1}^{K^c} x_i^c(t)\leq \Lambda^c
	\label{bwcon2}
	\end{equation}
	where $\Lambda^c$ denotes the number of packets that the service provider agrees to serve over a horizon $H$.
	
	
\subsection{Far-Sighted Constrained Resource Allocation using MPC}
In case of multiple classes, the objective \eqref{mpcobj} is summed over all $c\in C$. The state equations \eqref{stateeq} and constraints \eqref{bwcon2} are enforced for each class $c$. However, for constraint \eqref{bwcon1}, we would have a sum over all $c\in C$, since the total allocation over all classes in a slot cannot exceed the available bandwidth. The multi-class MPC problem can then be stated as:

	\begin{subequations} \label{mpcprob2}
		\begin{gather}
		\max_{x_{i}^c(t)} \;\; \sum_{c\in C}\sum_{t=0}^{H}\sum_{i = 1}^{K^c} x_i^c(t) - x_1^c(0) \quad \label{obj3} \\
		\textrm{s.t.} \;\; Q_{i-1}^c(t+1) = Q_{i}^c(t)-x_i^c(t) \; \forall c,i\; 1 < i \leq K^c \label{state111} \\
		Q_{K^c}(t+1) = A^c(t), \: \forall t,c 	\label{state211} \\
		x_i^c(t) \leq Q_i^c(t), \: \forall t,c \label{con111} \\
		\sum_{c\in C}\sum_{i = 1}^{K^c} x_i^c(t) \leq \Lambda, \forall \: t \label{con211} \\
		\sum_{t=0}^{H}\sum_{i = 1}^{K^c} x_i^c(t)\leq \Lambda^c, \: \forall \: c \label{con311}\\ 
		x_1^c(0) = \min\{\Lambda,Q_1^c(0)\}, \: \forall \:c \\
		x_i^c(t)\in \mathbb{N}\cup{0}  ~\forall t,c,i,~ 1 < i \leq K^c 
		\end{gather}
	\end{subequations}
	We now present an efficient solution to the MPC problem which would be beneficial for real-time systems.
	\begin{figure}
	\centering
		\includegraphics[width=0.95\linewidth]{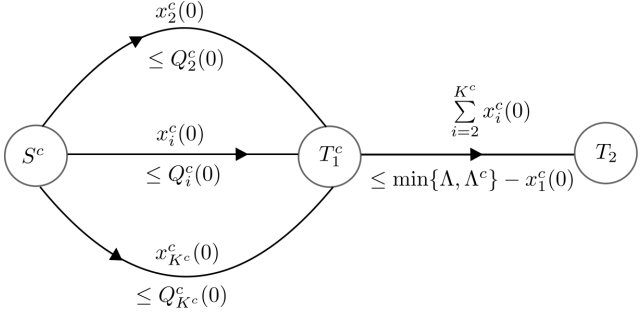}
		\caption{Shows the max-flow equivalent of the optimization}
		\label{proof1}
	\end{figure}
	
	\section{Methodology and Theoretical Analysis}
    \label{solmethod}
	The solution to the MPC problem \eqref{mpcprob2} requires solving an ILP due to integral restrictions of the variables $x_c^i(t)$. To this purpose, we show two important aspects of the problem. First, we identify that the short-sighted version of the problem with $H=0$ can be cast as a max-flow problem \cite{papadimitriou1998combinatorial}. This is useful for implementations in PON without any prediction. Next, we show that the ILP can be solved by its equivalent linear program that is obtained by relaxing the integral constraints on the variables. Thus, for solving \eqref{mpcprob2}, we have the following two propositions:
	
	\begin{proposition}
		The optimization in \eqref{mpcprob2} can be cast as a max-flow problem for $H = 0$ and can be solved in $\mathcal{O}(|C|\sum_cK^c)$. This is the optimal strategy for bandwidth allocation without prediction.
	\end{proposition}	
	\begin{proof}
        The short-sighted version of \eqref{mpcprob2} is obtained for $H=0$. We first show this for a single class $c$ and can be easily extended for multi-class scenarios.
        \begin{subequations} \label{myo1}
		\begin{gather}
		\max \; \sum_{i = 2}^K x_i^c(0) \label{onj}\\
		\textrm{s.t.} \;\; 
		x_i^c(0) \leq Q_i^c(0) \label{con1} \\
		\sum_{i = 2}^{K^c} x_i^c(0)\leq \min\{\Lambda,\Lambda^c\} - x_1^c(0)
		\label{con2} \\
		x_i^c(0)\in \mathbb{N}\cup \{0\}
		\end{gather}
	    \end{subequations}

		Consider a graph $G$ with a source node $S^c$ and two terminal nodes $T^c_1$ and $T^c_2$ respectively as shown in Fig. \ref{proof1}. Clearly, the constraints in \eqref{myo1} can be seen as the capacities of the links and the maximum flow from the source $S$ to terminal $T^c_2$ solves the required optimization. It is known that the complexity of solving max-flow problem for a graph with $V$ vertices and $E$ edges is $\mathcal{O}(VE)$ \cite{orlin2013max}.The graph consists of $3$ vertices and $K^c-1$ edges. Hence, the complexity of solving the short-sighted case is $\mathcal{O}(K^c)$. In case of multi-class, we will have $|C|$ such graphs corresponding to each class. The source $S^c$ for each class is connected to a single source node $S$. $T_2^c$ is replaced by a single terminal $T_2$ where all links from $T_1^c$ for each class terminate. Finally a terminal node $T$ is required to implement \eqref{con211} where the link from $T_2$ to $T$ has the capacity $\Lambda - \sum_c x_1^c(0)$. Thus the multi-class short-sighted allocation can be solved in $\mathcal{O}(|C|\sum_cK^c)$
	\end{proof}
	\begin{proposition}
		The MPC problem \eqref{mpcprob2} can be solved by its equivalent linear program obtained by relaxing the integral constraints.
	\end{proposition}	
    \begin{proof}
    To show this, we use an established result in the field of integer linear programming which provides a sufficient condition for an ILP to yield same solution as its linear program version \cite{bertsekas1997nonlinear}. A linear program of the form $\max\{\mathbf{cx}|\mathbf{Ax}\leq \mathbf{b}\}$ has integral solutions if the matrix $\mathbf{A}$ satisfies a special property called total unimodularity (TU), provided that $\mathbf{b}$ is integral. We show that the constraint matrix formed in our problem is TU. Since all bounds on our variables ($Q_i^c(t),\Lambda,\Lambda^c$) are integers, the corresponding LP yields integral solutions.  The details of the proof and relevant results required for the same provided in Appendix of the paper.
    \end{proof}

\textbf{Generic appeal of the solution:} We briefly mention some salient features of our solution technique. It is evident that the MPC problem of \eqref{mpcprob2} can be cast into an LP equivalent if the constraints are formed by adding the flow variables. Thus, any constraint of the form $\sum x_c^i(t)\lesseqgtr\alpha$ where $\alpha$ is a known constant may be realized. The sum can be carried over any combination of the variables $c$, $i$ or $t$. Additionally one might be interested in considering allocations for each user separately instead of the aggregated demand considered here. The corresponding optimizer would then store each user request separately and introduce more granularity in defining the variables $x^c_i$ as $x^c_{i,u}$ for an user $u$. Even then, the optimization can be performed by considering constraints for each user. Hence, the applied solution technique is quite powerful and allows efficient implementation for a large class of bandwidth and delay-constrained problems. 

	\section{Implementation in PON framework}
	\label{impl}
	We discuss the details of implementing the obtained allocation policies by solving \eqref{mpcprob2}.

	\begin{figure*}[t]
		\centering
		\includegraphics[scale=0.45]{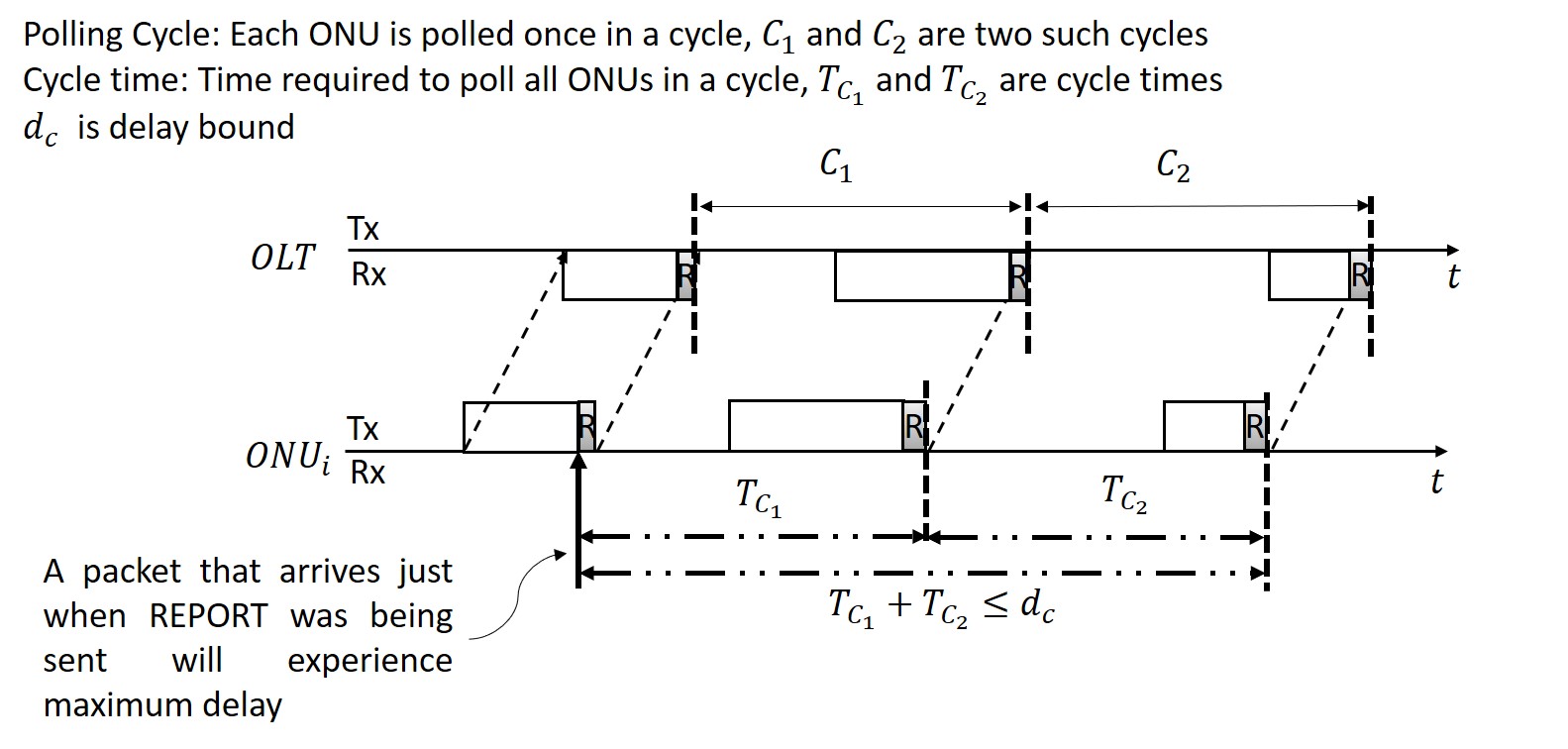}
		\caption{Illustration of Constraint from dynamic bandwidth allocation protocol}
		\label{fig:dba}
	\end{figure*}
	\subsection{Dynamic Bandwidth Allocation at fog node}
	\label{dbaproto}
	In order to implement the resource allocation scheme obtained from MPC at a fog Node in PON, the FN should be facilitated with an access protocol. As mentioned here, the fog node can implement a polling based scheme similar to the OLT. In this regard, the specifics of such a protocol are to be decided. We consider that the fog node can poll the ONUs with help of control messages similar to GATE and REPORT as per Multi Point Control Protocol (MPCP) \cite{ieeestd}. Each ONU sends a REPORT message containing information regarding its queue status for each class of service. The fog node collects this information from all ONUs and obtains the resource allocation scheme by solving \eqref{mpcprob2}. It is to be noted that the solution of \eqref{mpcprob2} provides the aggregated (combined for all ONUs) number of packets/bits to be cleared for each class. Thus, the obtained solution is to be distributed among the ONUs. Since the requests for each ONU is known, a max-min fair allocation (also known as excess distribution) \cite{le2005rate} of the aggregated amount is granted to each ONU for each class. The information regarding the granted amount is sent via a GATE message. The specific details of the cycle time, slot duration considered for optimization and horizon for optimization are discussed below:
	
	\subsection{Constraint on Cycle Time}
	
	\label{cycle}
	A cycle is a duration of time in which all ONUs are served once as shown in Fig.  \ref{fig:dba}. After the data transmission, each ONU generates and appends a REPORT message which contains information regarding the current buffer status for each class. It is evident that the cycle time is indeed dependent on the amount of data that the ONUs are allowed to clear on receiving the GATE message. Usually heuristic allocation strategies implement a limit on the amount of granted bandwidth for each ONU (otherwise a malicious ONU might occupy the entire bandwidth) which effectively limits the cycle time. Let us denote the maximum allowable grant for each ONU by $T_{mg}$ in seconds.  Then, the maximum amount of time within which all ONUs are served once, is given by $T_m = N T_{mg}$ (neglecting the GATE, REPORT and guard durations). It is easy to see that given a delay constraint say $d_c$, it is not possible for cycle time $T_m$ to be arbitrarily large. To estimate the limit on the cycle time. we observe the worst possible delay that can be encountered by a packet. A packet that arrives just after the data has been transmitted has to wait for at least a cycle to get reported. Then, if required (depending on the delay constraint) this packet can be cleared in the subsequent cycle. This implies that the maximum allowable cycle time cannot be more than $\frac{d_c}{2}$. Otherwise, two consecutive cycles might be more than $\frac{d_c}{2}$ which would lead to a delay violation for the worst packet as described here (see Fig. \ref{fig:dba}). Hence,
	\begin{equation}
	T_m \leq \frac{d_c}{2}
	\label{cyl}
	\end{equation}
	\subsection{Optimal Cycle time and choice of Time slot duration}
	The optimization of \eqref{mpcprob2}, is solved at end of each time slot. As evident from \eqref{minstep}, the minimum duration for time slot can be $\max\{rtt_i^f \forall i\}$ while the maximum will be that of the maximum cycle time that is $\frac{d_c}{2}$. It is important to note that while defining the cycle we had ignored the guard durations and the control messages. To minimize wastage of bandwidth, it is evident that the number of these messages and guard durations must be minimal. Thus, a good choice for time slot duration should be the maximum cycle time i.e., $\frac{d_c}{2}$. This implies that AoI for IoT devices would be bounded by $d_c$.
	\begin{equation}
	T_s = \frac{d_c}{2}
	\end{equation}  
	
	
	\section{Simulation Results and Discussion}
	\label{sim}
	We compare the performance of our (MPC-based) protocol with 1) existing standards and 2) the two  protocols discussed in Section \ref{relpon}: (i) OOB protocol \cite{comparewith}  and (ii) priority protocol \cite{ou2017resource}.
We compare the performance of these protocols using three metrics: (i) percentage of violation/drop, (ii) throughput, (iii) average delay, and (iv) jitter.
	The decentralized out of band bandwidth allocation which employs separate wavelengths is referred as "OOB". The one in \cite{ou2017resource} is referred as "Priority based" and our protocol is referred as "MPC". The simulation parameters have been tabulated in Table \ref{tab2}. Note that the priority based allocation scheme can also be conceived as a heuristic based on current information or a short-sighted allocation scheme. We have simulated a self-similar traffic with 16 traffic sources attached to each ONU in OMNet++ \cite{Varga01theomnet++}. The self-similar traffic has been simulated with Hurst parameter of 0.8 \cite{IPACT} to generate background best-effort traffic. Our choice of delay and bandwidth requirements are $(d_1,b_1) = (1 \text{ ms}, 100 \text{ Mbps})$ and $(d_2,b_2) = (4 \text{ ms}, 100 \text{ Mbps})$ respectively (our model is generic and can handle any delay and bandwidth requirement unless physically infeasible).  Owing to the least delay constraint, the time slot was chosen to be $0.5$ ms. This gives $K^1 = 1$ and $K^2 = 7$. We choose a finite horizon $H = 10$. To generate the traffic corresponding to delay-constrained applications, a Hurst parameter of 0.2 was considered.
	To solve the optimization, we used MATLAB2020a and integrated OMNet++ to work seamlessly with MATLAB. 
	\begin{table}
		\begin{center}
			\caption{Simulation Parameters}
			\label{tab2}
			\begin{tabular}{|l|l|}
				\hline
				\multicolumn{1}{|c|}{\textbf{Parameter}} &  \multicolumn{1}{|c|}{\textbf{Value}} \\ 
				\hline
				Traffic Generator & Self Similar using Pareto \cite{IPACT}  \\ 
				\hline
				Hurst Parameters & 0.8 and 0.2 \cite{IPACT}  \\
				\hline
				Buffer Size for each ONU & 10 Mb  \\ 
				\hline
				Link rate at ONU side & 100 Mbps \\
				\hline
				PON link rate &  1 Gbps  \\ 
				\hline
				Number of ONUs &  16 \\
				\hline
				Traffic Sources per ONU & 16  \\ 
				\hline
				Maximum Cycle Time & 0.5 ms \\
				\hline
				Guard Duration & 5 $\mu s$ \cite{IPACT, an2014new}   \\ 
				\hline
				Distance between OLT and fog node & 15 km  \\
				\hline
				Distance between ONUs and fog node & 1-5 km \cite{comparewith}   \\ 
				\hline
				$(d_1,b1)$ & $(1 ms , 100 Mbps)$\cite{fitzek2020computing}\\
				\hline
				$(d_2,b2)$ & $(4 ms , 100 Mbps)$\cite{fitzek2020computing}\\
				\hline
				Confidence Interval & 95\% \\
				\hline
			\end{tabular}
		\end{center}
	\end{table}
	
	
	\subsection{Comparison based on Percentage of Dropped Packets}
	Here, we present two types of comparisons. First, we present the comparison based on existing standards as discussed in Section \ref{intro}. Then, we compare our proposal with some heuristics which were proposed for accommodating new applications.
	\subsubsection{Comparison with existing standards}
	\begin{figure}
	    \centering
	    \includegraphics[width=0.8\linewidth]{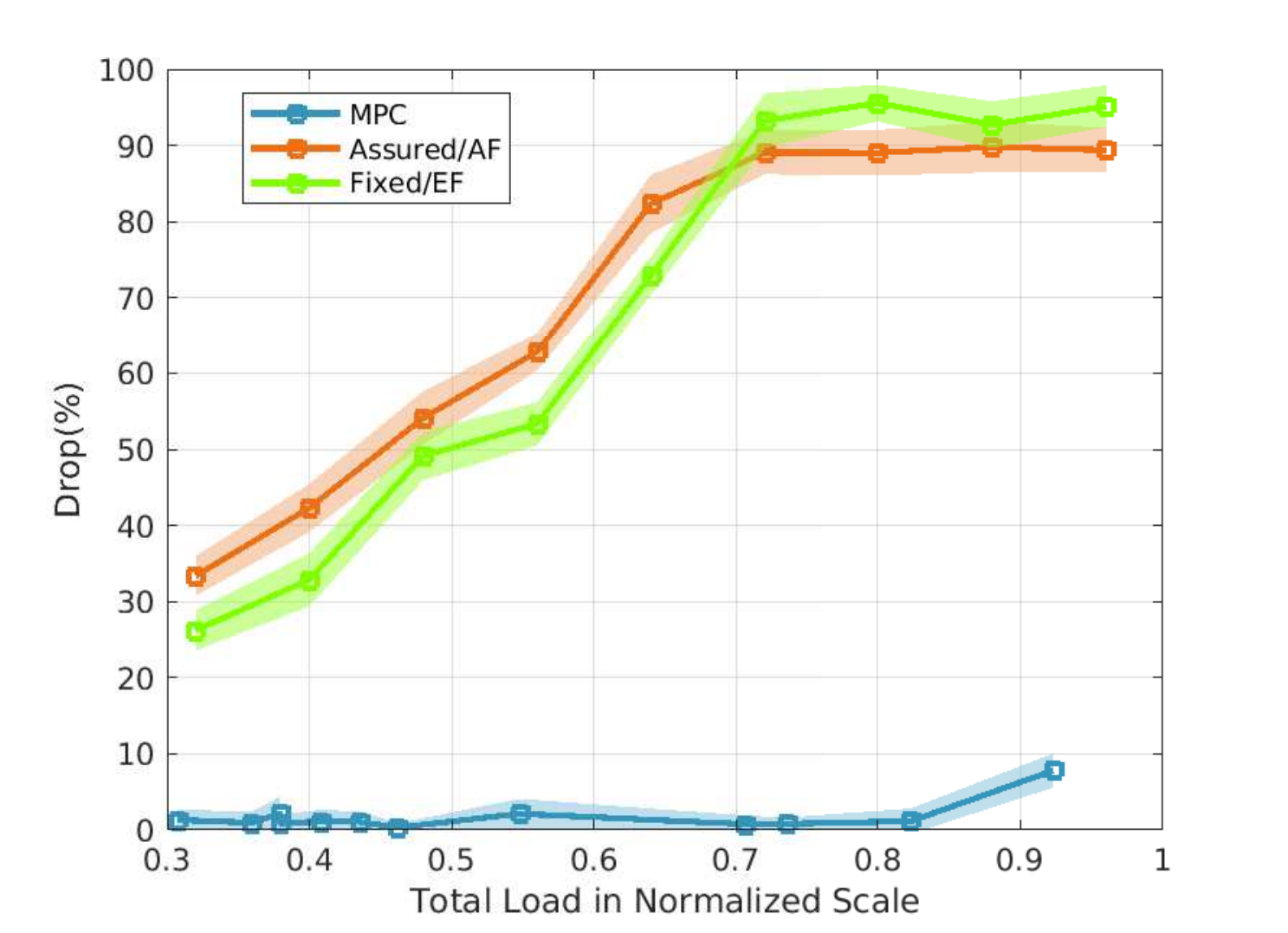}
	    \caption{Drop percentage comparison for delay constraint of 1ms with Fixed Allocation and Assured Allocation policies.}
	    \label{fig:stds}
	\end{figure}
	Fig. \ref{fig:stds} shows the performance of existing standards which do not cater to satisfying strict delay constraints for bursty traffic conditions. The drop percentage refers to those packets which exceed the delay constraint of 1ms. Evidently, MPC based proposal specifically does not allow drops unless the load crosses a threshold where satisfying the delay constraint becomes impossible due to bandwidth restrictions. On the other hand, Fixed allocations which are primarily designed for constant-bit rate applications, allocate a fixed portion of bandwdith (TDM) for each ONU. Due to the bursty nature of arrivals, this scheme is highly inefficient (specifically at low loads, as it unnecessarily reserves bandwidth which delays subsequently arriving bursts). Assured scheme shows better performance as compared to Fixed allocations only at lower loads as it gets opportunity to dynamically schedule the bursty packet arrivals. However, Assured strategy only guarantees a certain bandwidth and does not aim at maintaining delay constraints. As the load increases, the two strategies show similar performance since, the fixed and assured allocations reach their respective limits.
	\subsubsection{Comparison with existing heuristics}
	\begin{figure*}
		\centering
		
		\begin{subfigure}[b]{0.3\textwidth}
			\centering
			\includegraphics[width=\textwidth]{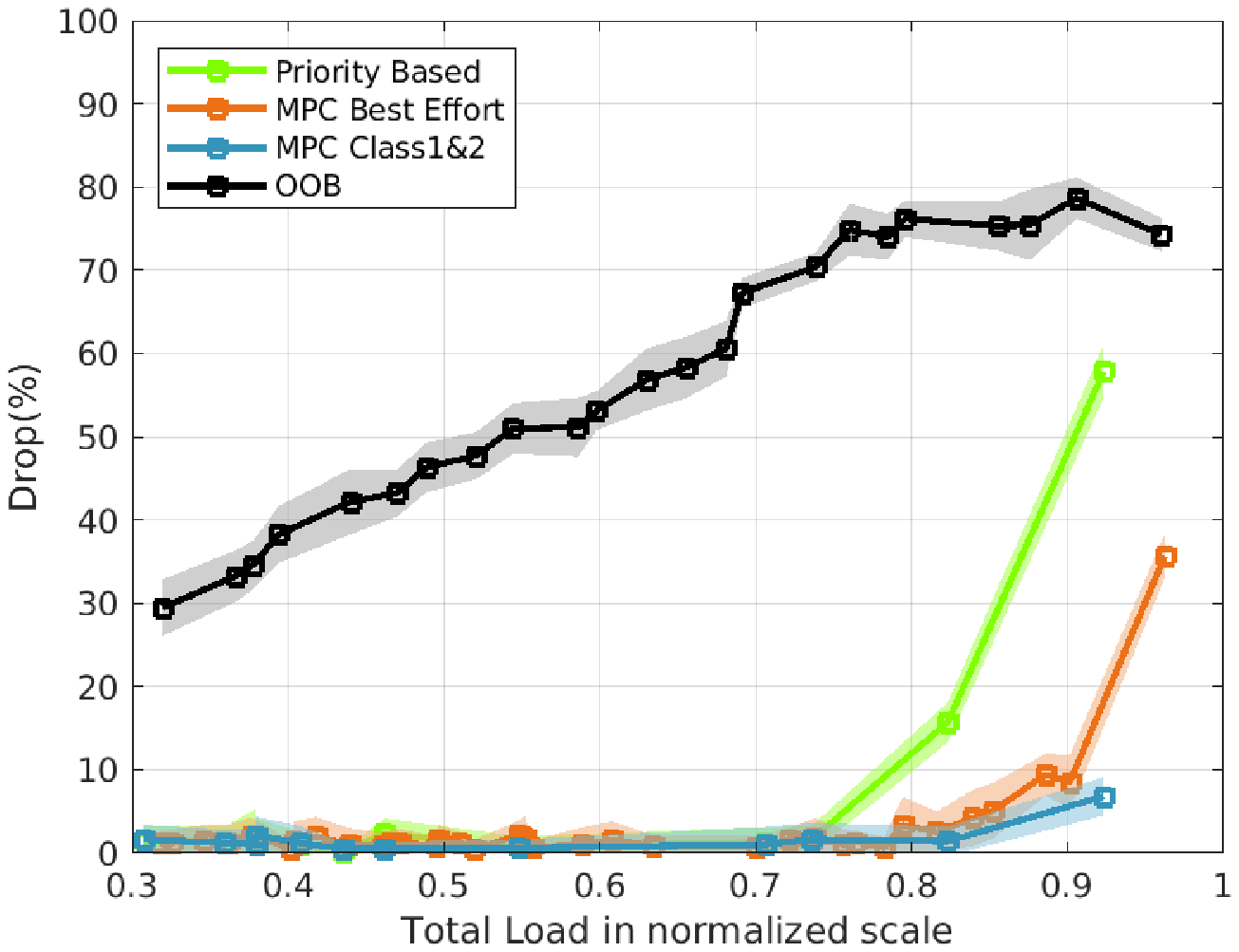}
			\caption{Drop percentage for all classes of traffic}
			\label{fig:drop25}
		\end{subfigure}
		\hfill
		\begin{subfigure}[b]{0.3\textwidth}
			\centering
			\includegraphics[width=\textwidth]{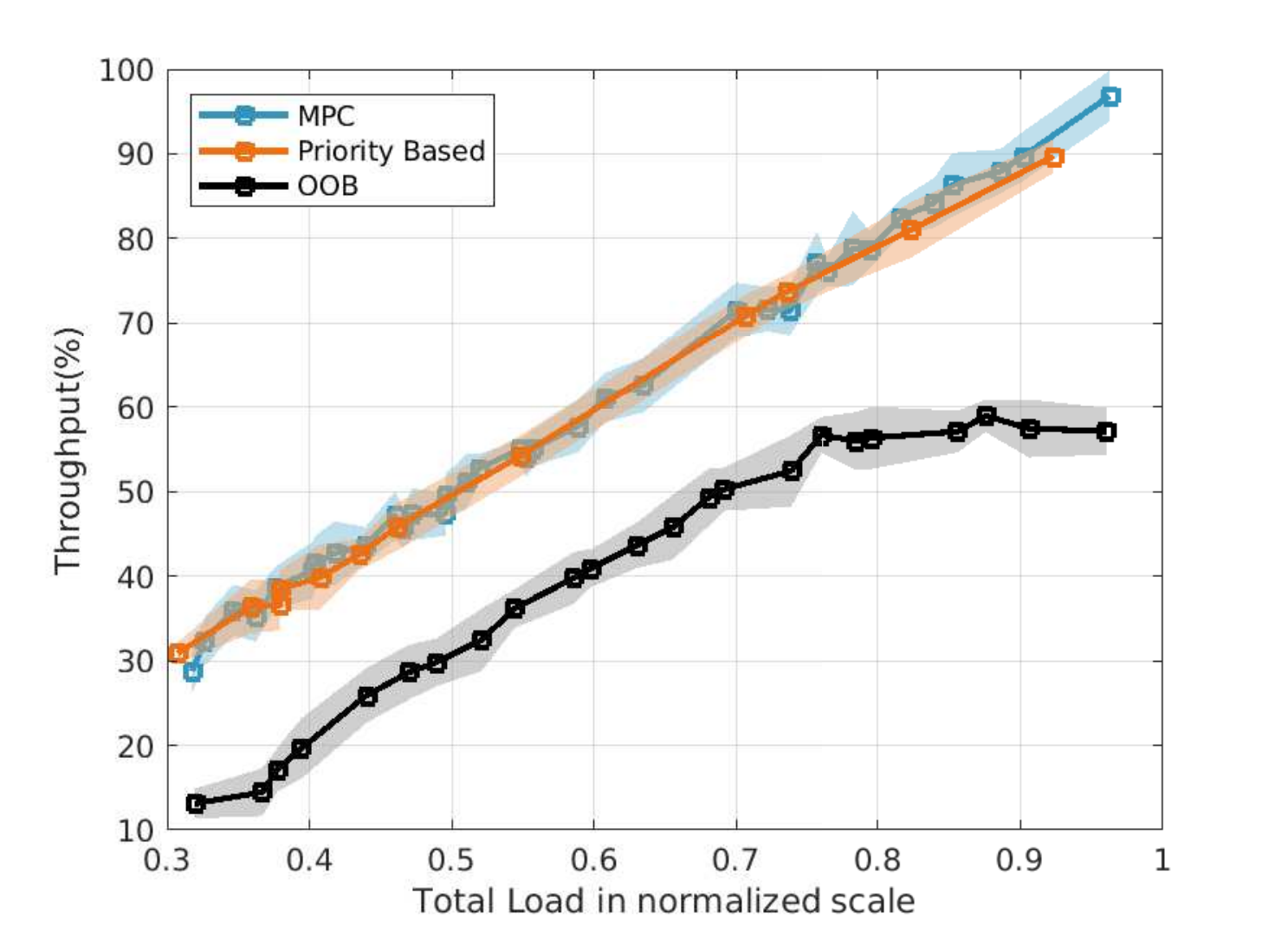}
			\caption{Throughput\%}
			\label{fig:through25}
		\end{subfigure}
		\hfill
		\begin{subfigure}[b]{0.3\textwidth}
			\centering
			\includegraphics[width=\textwidth]{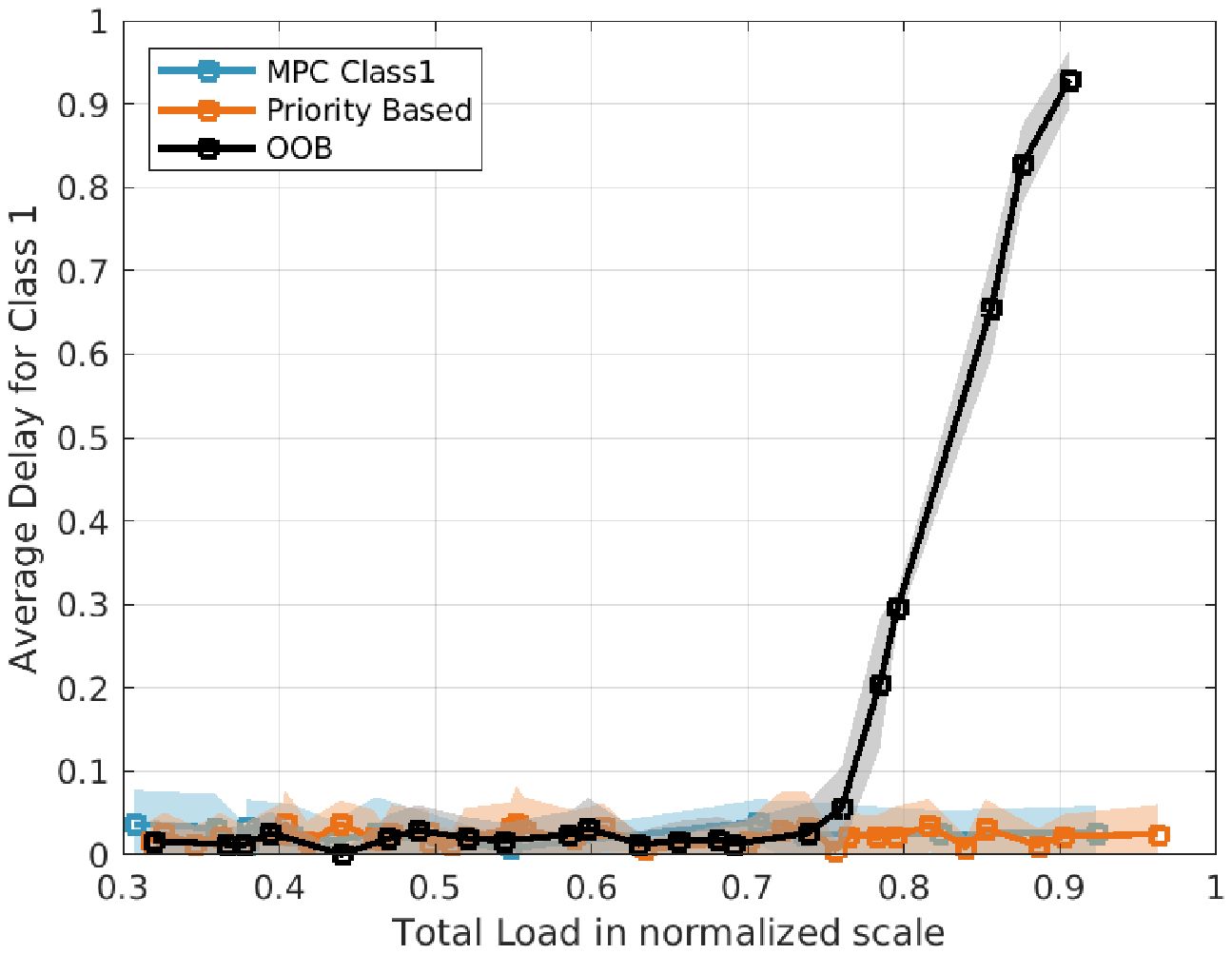}
			\caption{Average Delay for class 1}
			\label{fig:delay1_25}
		\end{subfigure}
		\hfill
		\begin{subfigure}[b]{0.3\textwidth}
			\centering
			\includegraphics[width=\textwidth]{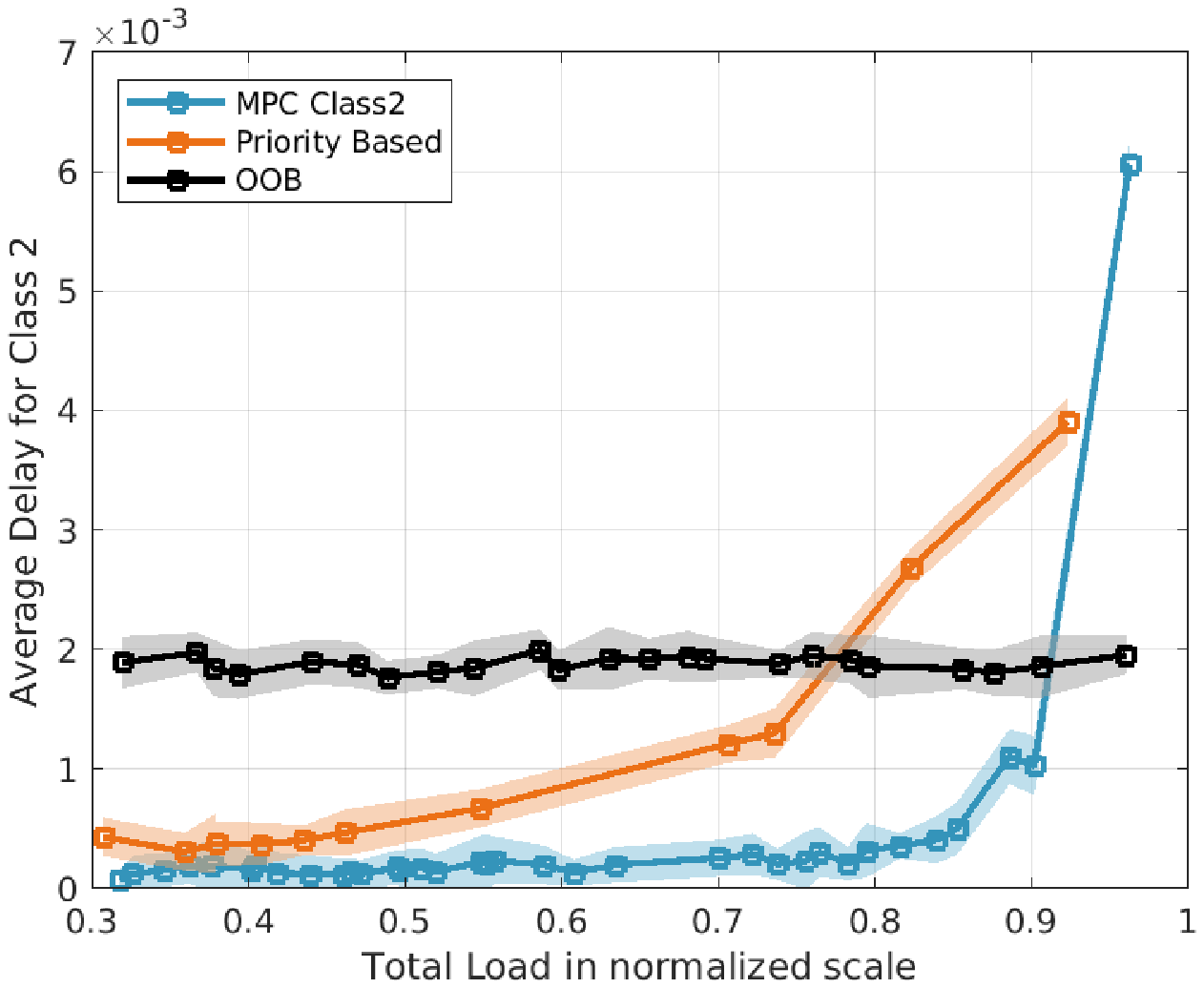}
			\caption{Average Delay for class 2}
			\label{fig:delay225}
		\end{subfigure}
		\hfill
		\begin{subfigure}[b]{0.3\textwidth}
			\centering
			\includegraphics[width=\textwidth]{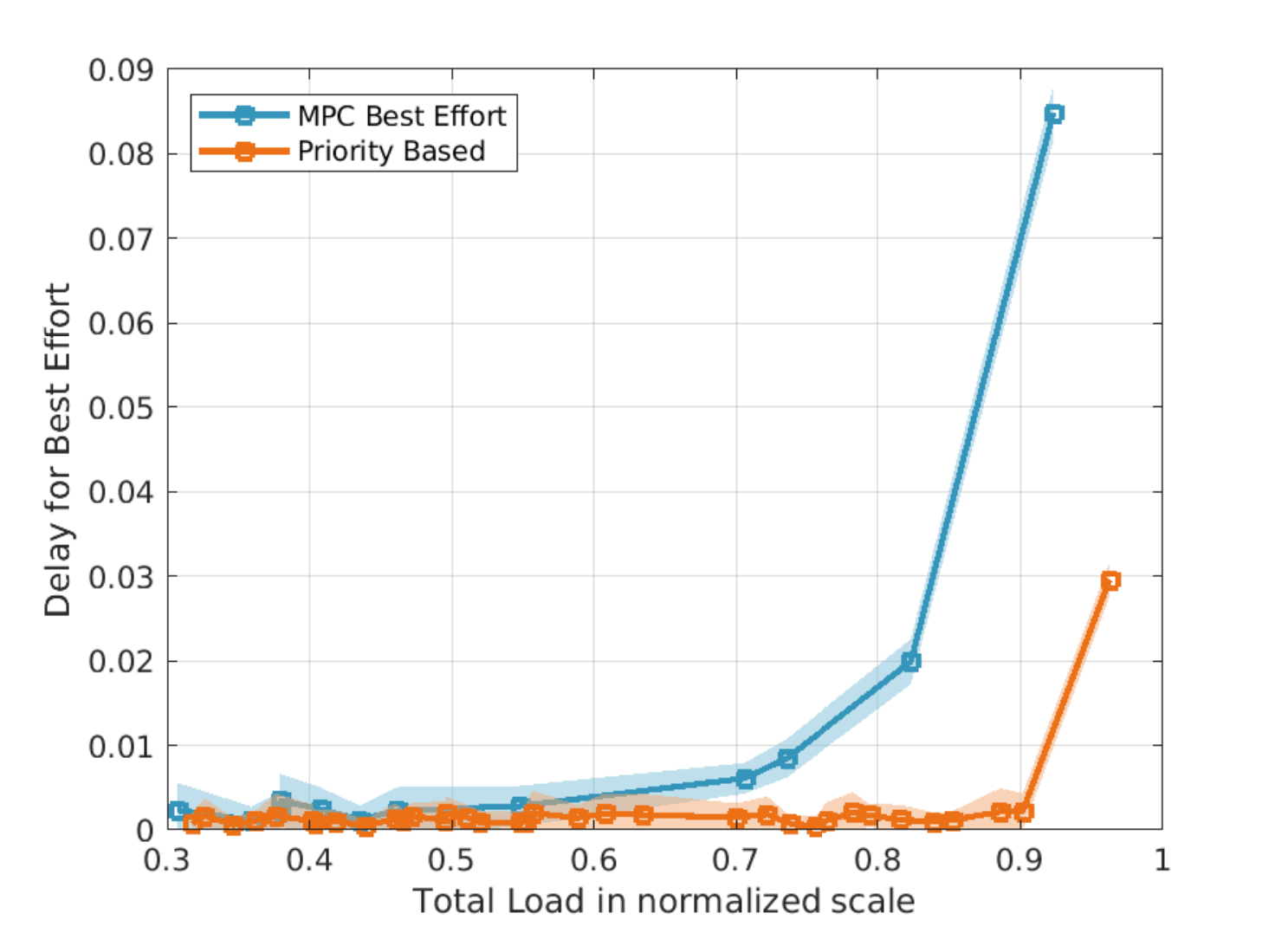}
			\caption{Average Best Effort Delay}
			\label{fig:delaybe25}
		\end{subfigure}
		\hfill
		\begin{subfigure}[b]{0.3\textwidth}
			\centering
			\includegraphics[width=\textwidth]{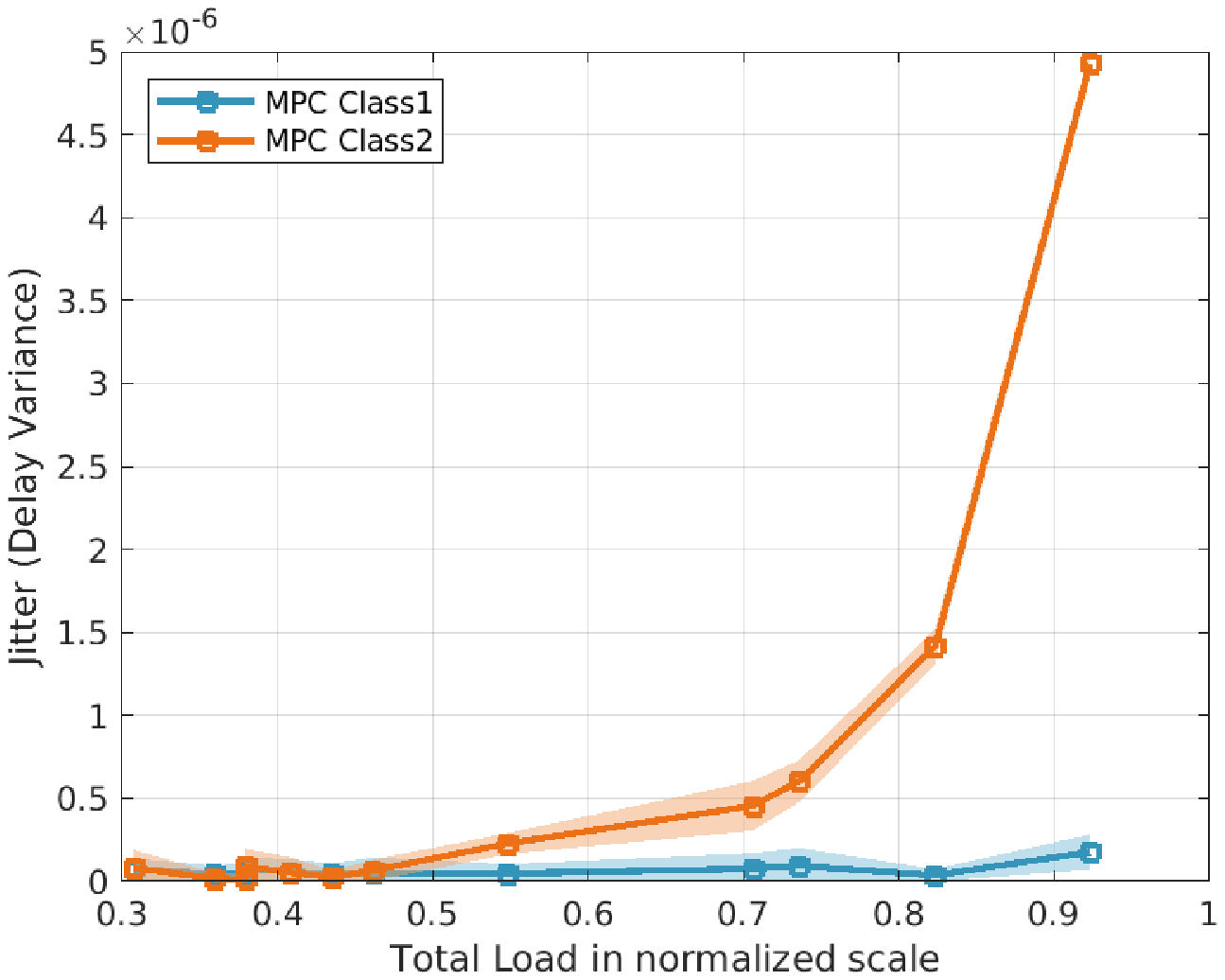}
			\caption{Jitter Performance for classes 1 \& 2}
			\label{fig:jitter25}
		\end{subfigure}
		\caption{Figure shows performance in terms of Drop, Throughput, Average Delay and Jitter with Best-effort Load of $25\%$ of the maximum load. Drop and Delay of Classes are maintained at cost of Best-effort delay}
		\label{fig:25}
	\end{figure*}
	\begin{figure*}
		\centering
		\begin{subfigure}[b]{0.3\textwidth}
			\centering
			\includegraphics[width=\textwidth]{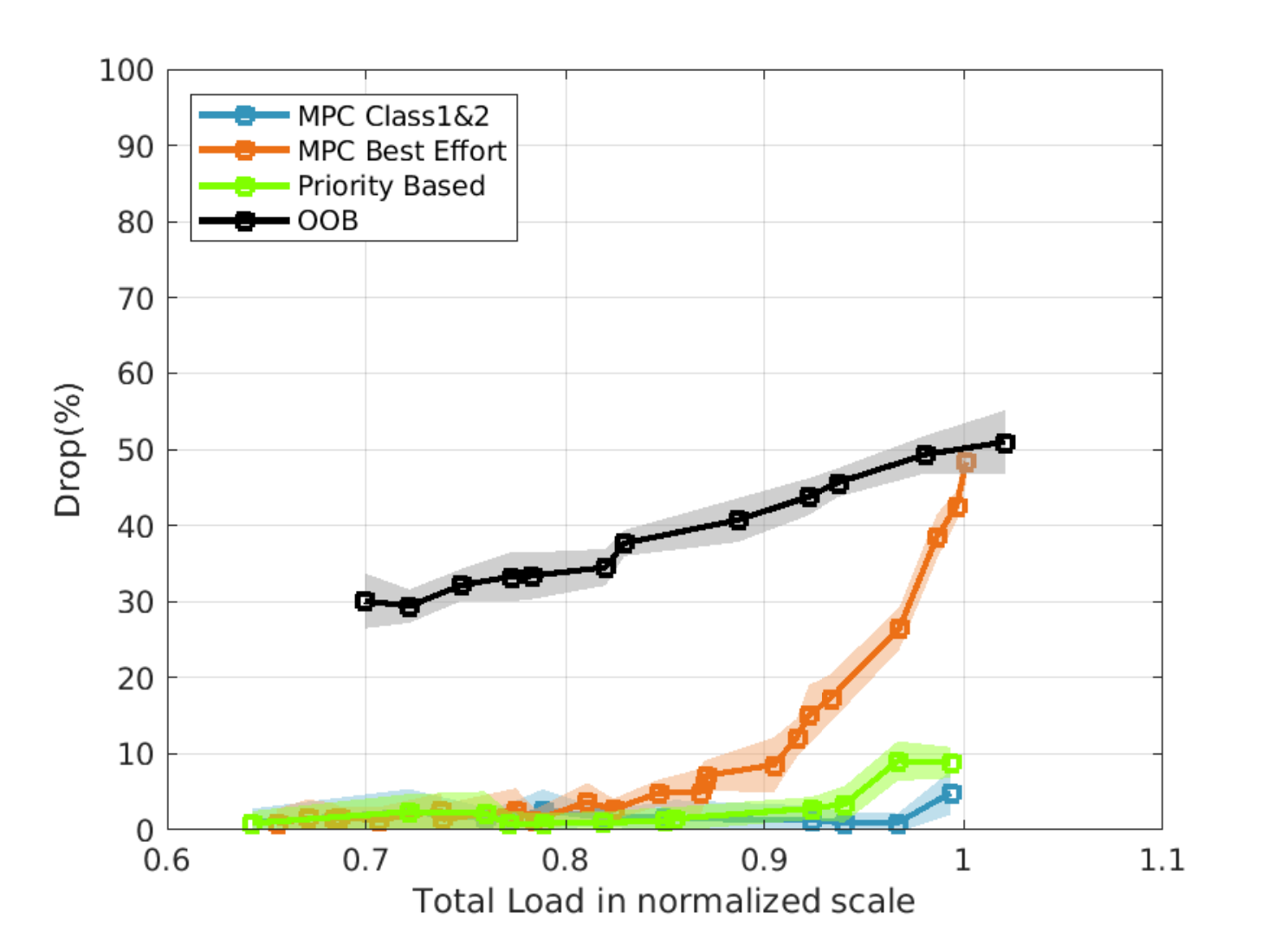}
			\caption{Drop percentage for all classes of traffic}
			\label{fig:drop50}
		\end{subfigure}
		\hfill
		\begin{subfigure}[b]{0.3\textwidth}
			\centering
			\includegraphics[width=\textwidth]{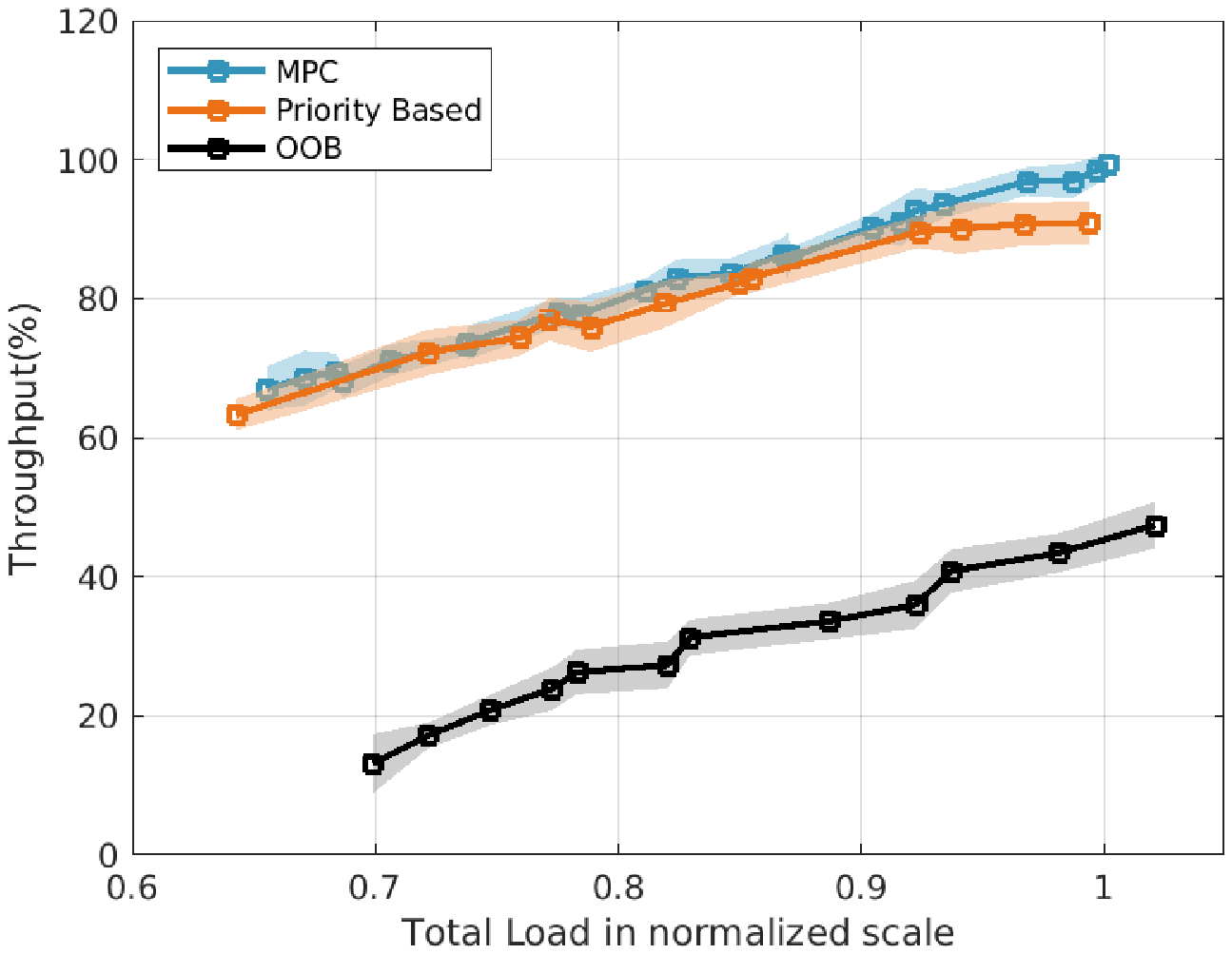}
			\caption{Throughput\%}
			\label{fig:through50}
		\end{subfigure}
		\hfill
		\begin{subfigure}[b]{0.3\textwidth}
			\centering
			\includegraphics[width=\textwidth]{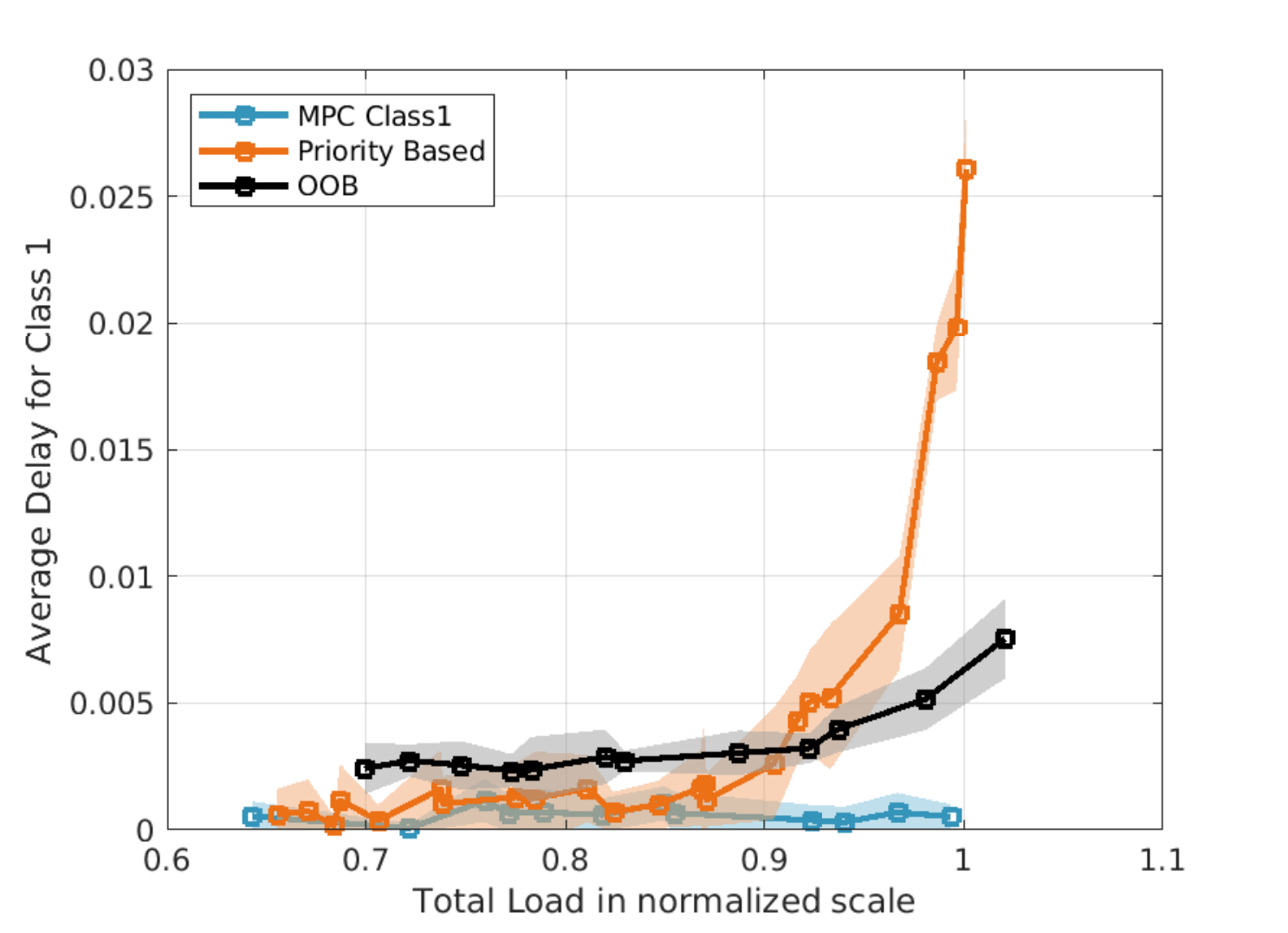}
			\caption{Average Delay for class 1}
			\label{fig:delay150}
		\end{subfigure}
		\hfill
		\begin{subfigure}[b]{0.3\textwidth}
			\centering
			\includegraphics[width=\textwidth]{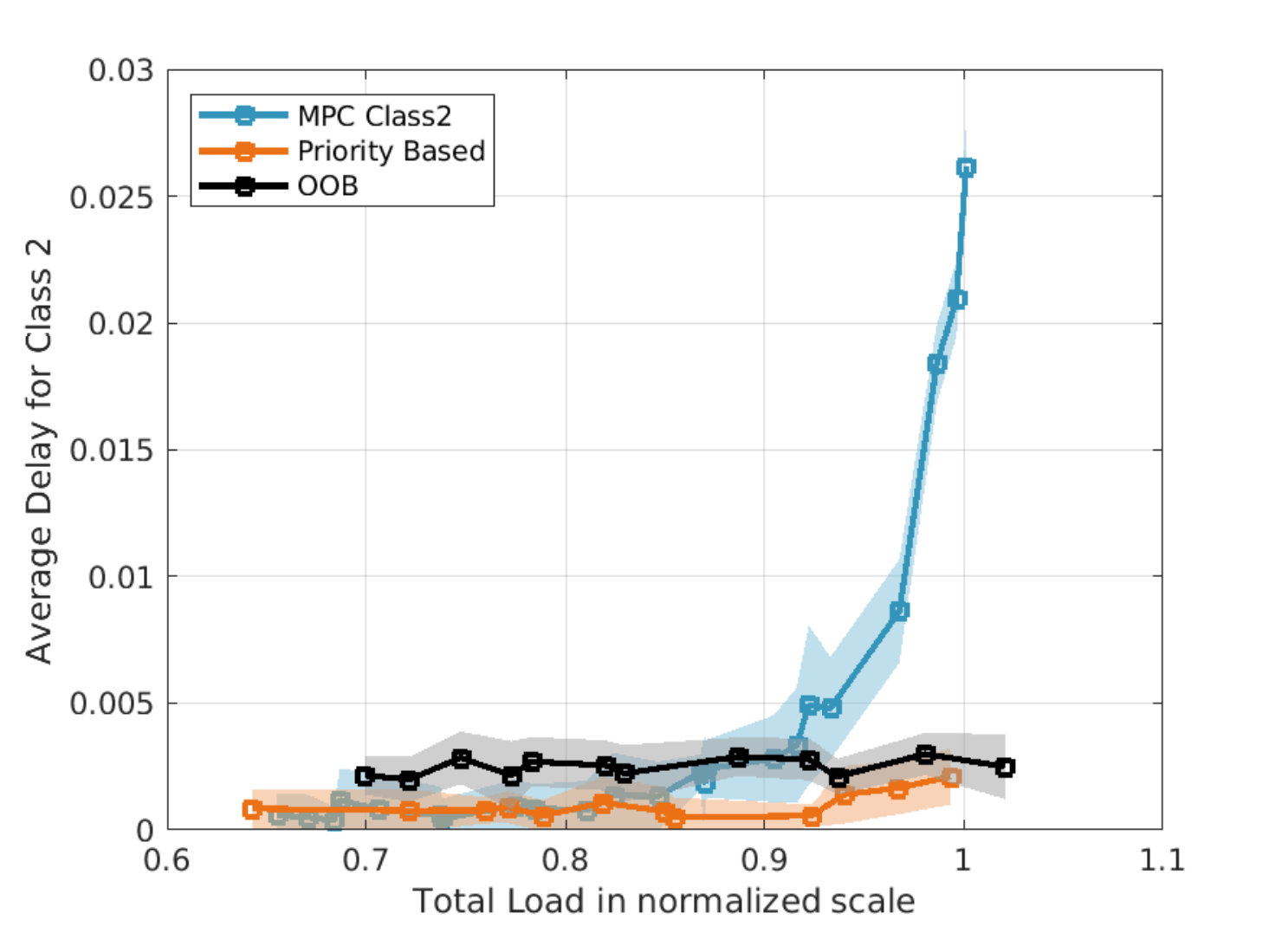}
			\caption{Average Delay for class 2}
			\label{fig:delay250}
		\end{subfigure}
		\hfill
		\begin{subfigure}[b]{0.3\textwidth}
			\centering
			\includegraphics[width=\textwidth]{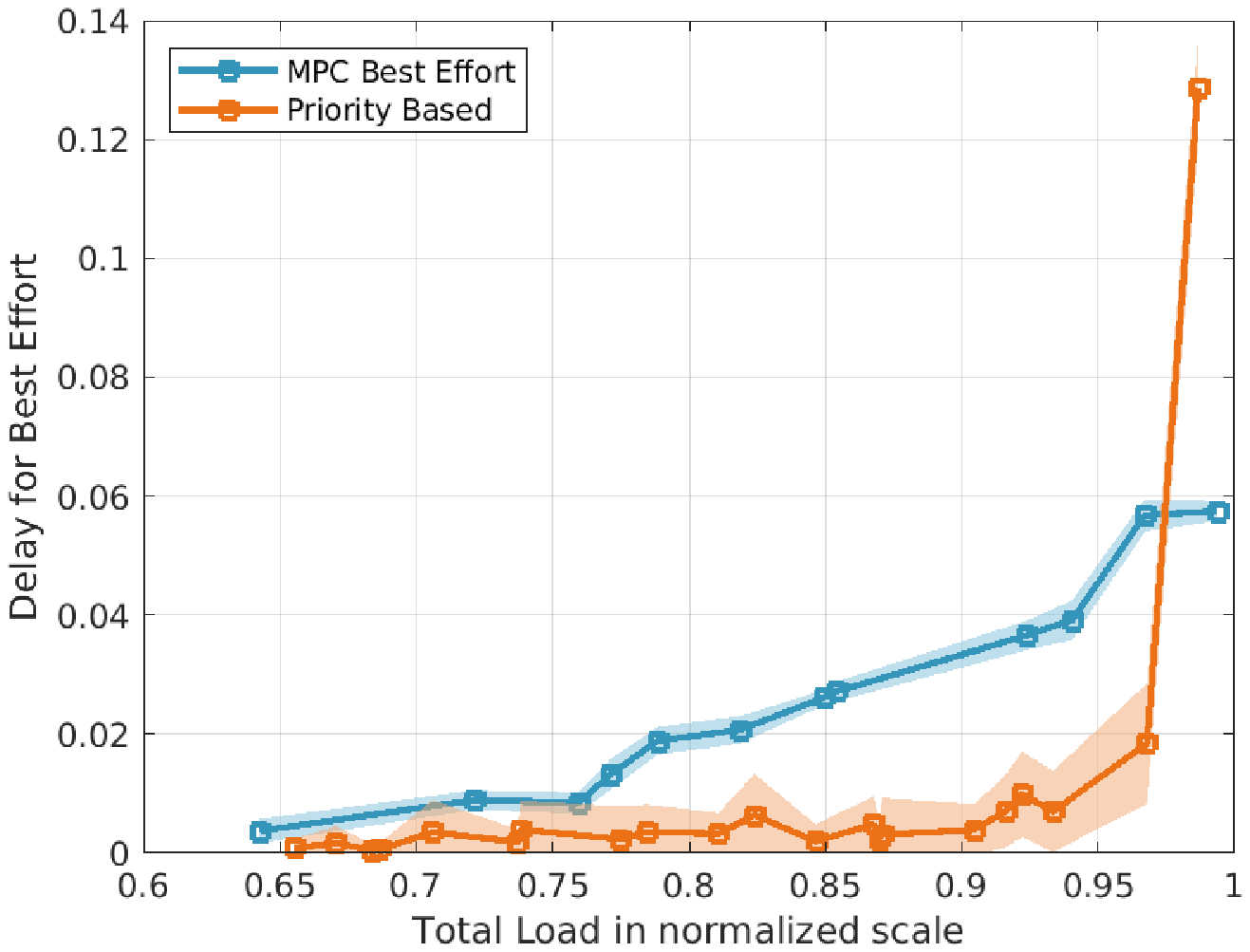}
			\caption{Average Best Effort Delay}
			\label{fig:delaybe50}
		\end{subfigure}
		\hfill
		\begin{subfigure}[b]{0.3\textwidth}
			\centering
			\includegraphics[width=\textwidth]{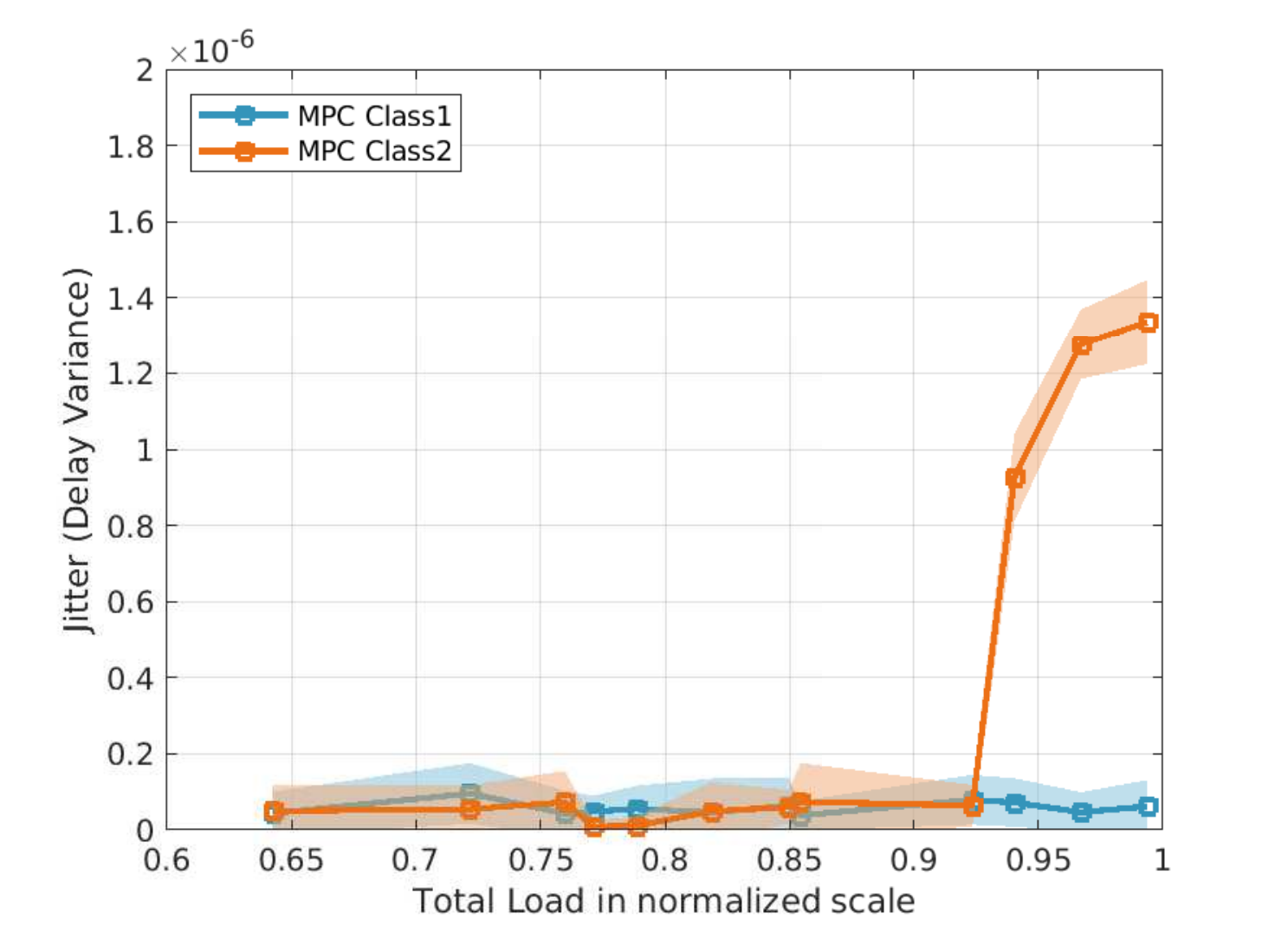}
			\caption{Jitter Performance for classes 1 \& 2}
			\label{fig:jitter50}
		\end{subfigure}
		\caption{Figure shows performance in terms of Drop, Throughput, Average Delay and Jitter with Best-effort Load of $50\%$ of the maximum load. Drop and Delay of Classes are maintained at cost of Best-effort delay. With the increase in best effort load, the average delay performance of class 2 suffers (our allocation is not delay optimal).}
		\label{fig:50}
	\end{figure*}
	
	To evaluate the performance of our proposal, we find the number of dropped packets with increase in aggregated loads of class 1 and class 2 services. We show the effect of increasing best-effort traffic as well. It is conceivable that the MPC based framework adheres to satisfying the strict QoS requirements first while best-effort traffic is not tackled by the concerned optimization. Best-effort traffic allocations are performed only if there is sufficient bandwidth after the optimal allocations obtained from \eqref{mpcprob2} are scheduled. Both OOB and priority based methods do not aim at guaranteeing the delay and bandwidth constraints and hence fail drastically when stringent real-time delay requirements are to be satisfied (as shown in Fig. \ref{fig:drop25}). On the other hand, the MPC based allocation strategy operates well within the feasible delay and bandwidth regions, unless the requests exceed the available bandwidth. Thus, as shown in Fig. \ref{fig:drop25}, MPC based allocations show almost zero drop. The best-effort drop (full queue size condition) is also plotted in the figure. We observe that the drop percentage for best-effort traffic is quite low because of the low percentage of best-effort traffic generated. With increase in class 1 and class 2 traffic, the percentage of drop for best-effort traffic increases.
	
	We can also observe the effect on increasing percentage of best-effort traffic. As shown in Fig. \ref{fig:drop50}, the best-effort drop increases significantly and this shows the trade-off for using MPC based allocations. MPC allocations try and ensure that the delay of class 1 and class 2 types of service do not get violated at the expense of dropping best-effort traffic.
	
	\subsection{Comparison based on Throughput}
	Througput is calculated as the percentage of time the link remains busy for clearing packets with respect to the total simulation time. Since MPC maximizes throughput as its objective, we observe its throughput performance when compared with the existing proposals. As expected, in Fig. \ref{fig:through25} and Fig. \ref{fig:through50}, the MPC having the optimal throughput shows the best performance. As pointed out earlier, a major drawback of OOB is that allocation for fog requests depend on the corresponding OLT allocation window. Hence, although there are two wavelenghts in use, utilization is drastically affected, as evident from Figs. \ref{fig:through25} and \ref{fig:through50}. On the contrary, the priority based scheme performs fairly well since it utilizes its available bandwidth with either the OLT requests or the prioritized fog requests. When the best-effort load increases, we start to observe that full throughput is achievable at lower loads for class 1 and class 2 services. This shows that even when the priority classes do not have traffic, the protocol does not suffer from light-load penalty, which is a common problem in priority based allocation \cite{IPACT}.
	\subsection{Comparison based on Average Delay}
	As discussed, we compare the average delay performance for the mentioned protocols. The average delay for both classes is drastically less compared to those for OOB and Priority based protocols. It is important to note that the objective of our MPC does not minimize average delay and was rather throughput based. However, the average delay performance for the most stringent class can be seen to outperform other counterparts as shown in Fig. \ref{fig:delay1_25} and \ref{fig:delay150}. The same cannot be said for the other classes as can be seen for some loads for the second class of service as shown in Fig. \ref{fig:delay225} and \ref{fig:delay225}. Also as evident from Fig. \ref{fig:delaybe25} and \ref{fig:delaybe50}, the best effort service suffers in terms of average delay when compared to that of a Priority based slicing protocol. It is important to note that the average delay performance is mostly by design and can be manipulated by adequate prioritization and handling of allocation decisions. When allocating class 1 and class 2 traffic from same virtual queue, we first allocate those of class 1 first and then allocate class 2. By changing this allocation strategy, we may obtain a desired average delay performance.
	
	We do not show the best effort performance for OOB since it is similar to that of class 2 because OOB does not handle classes separately but have differentiation between services for fog and OLT. We treated the traffic for class 1 as that intended for fog and all other classes get mapped to that for OLT.   
	\subsection{Jitter performance}
	Jitter can have many definitions. We use variance of delay as the measure for jitter. The jitter performance shows a peculiar behavior. The jitter performance for other protocols were drastically poor and could not be presented on the same scale (almost 10-100 times more). As shown in Fig.  \ref{fig:jitter25}, at low best-effort loads, the MPC based allocations show drastic increase in jitter for class 2 while class 1 experiences almost an uniform jitter performance. This might be due to the fact that at low loads, the buffer does not remain occupied at all times and hence delays mostly depend on the random arrival times. On the other hand, at higher best-effort loads, the buffers remain filled with more probability and hence jitter performance improves as is evident from Fig. \ref{fig:jitter50}.

	\subsection{Sensitivity Analysis}
		
		
	It is important to note that MPC requires future information (in this case arrivals) to obtain optimal allocations. Thus, implementation of MPC requires prediction of traffic from previous data. To evaluate its robustness, we introduce noise (Gaussian with zero mean and variance 25) to perfect prediction and look at the allocations obtained due to same. As shown in Fig. \ref{robust1} and \ref{robust2}, for both classes, the drop percentage remains well within 1\% (for class 1 it is within 0.1\%) which is quite desirable and exhibits the robust behavior of our algorithm .
	\begin{figure}[!ht]
		\centering
		\includegraphics[scale=0.5]{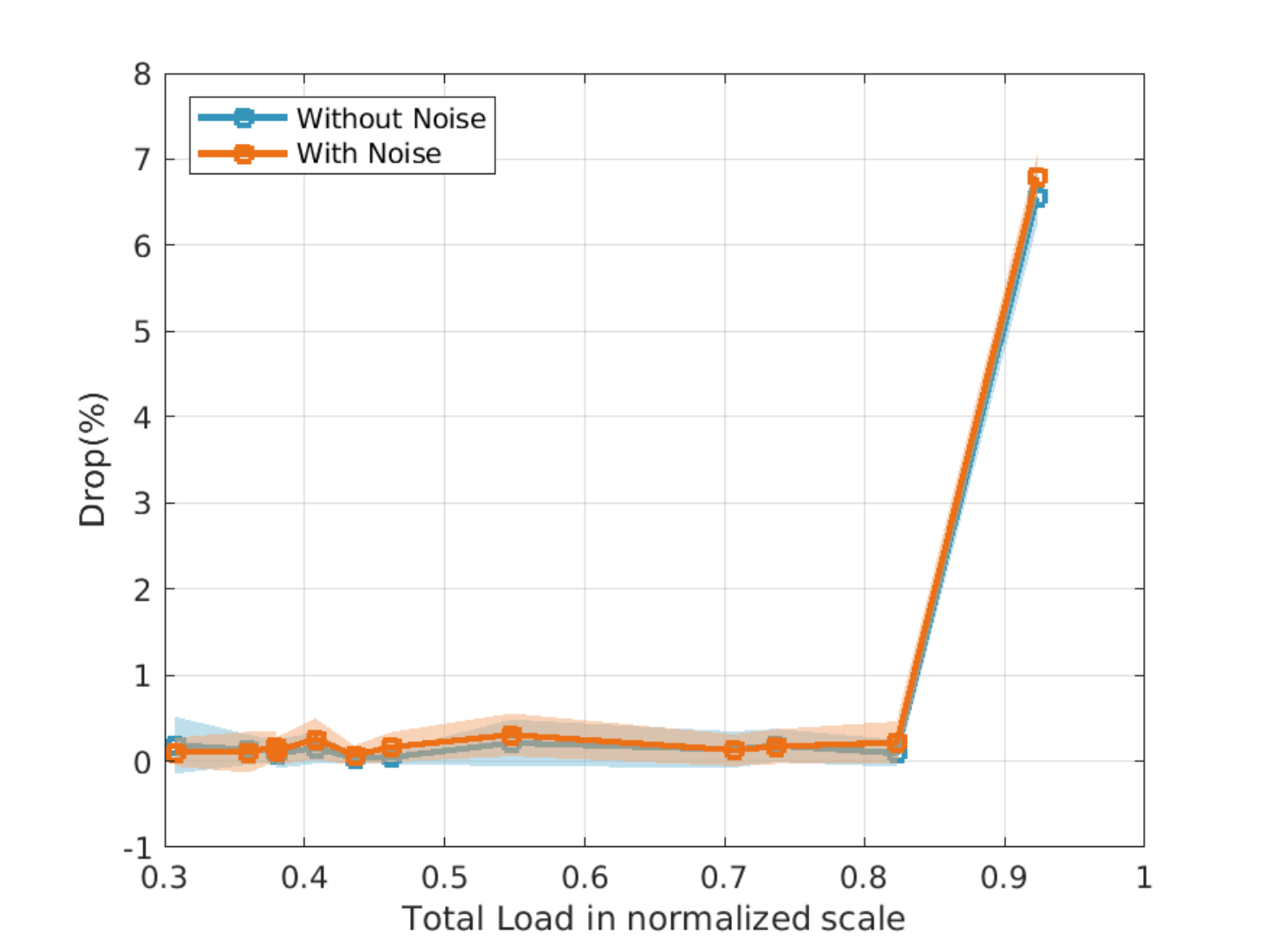}
		\caption{Drop percentage  for class 1 with and without noise. Noise is Gaussian with mean 0 and variance 25.}
		\label{robust1}
	\end{figure}
	\begin{figure}[!ht]
		\centering
		\includegraphics[scale=0.5]{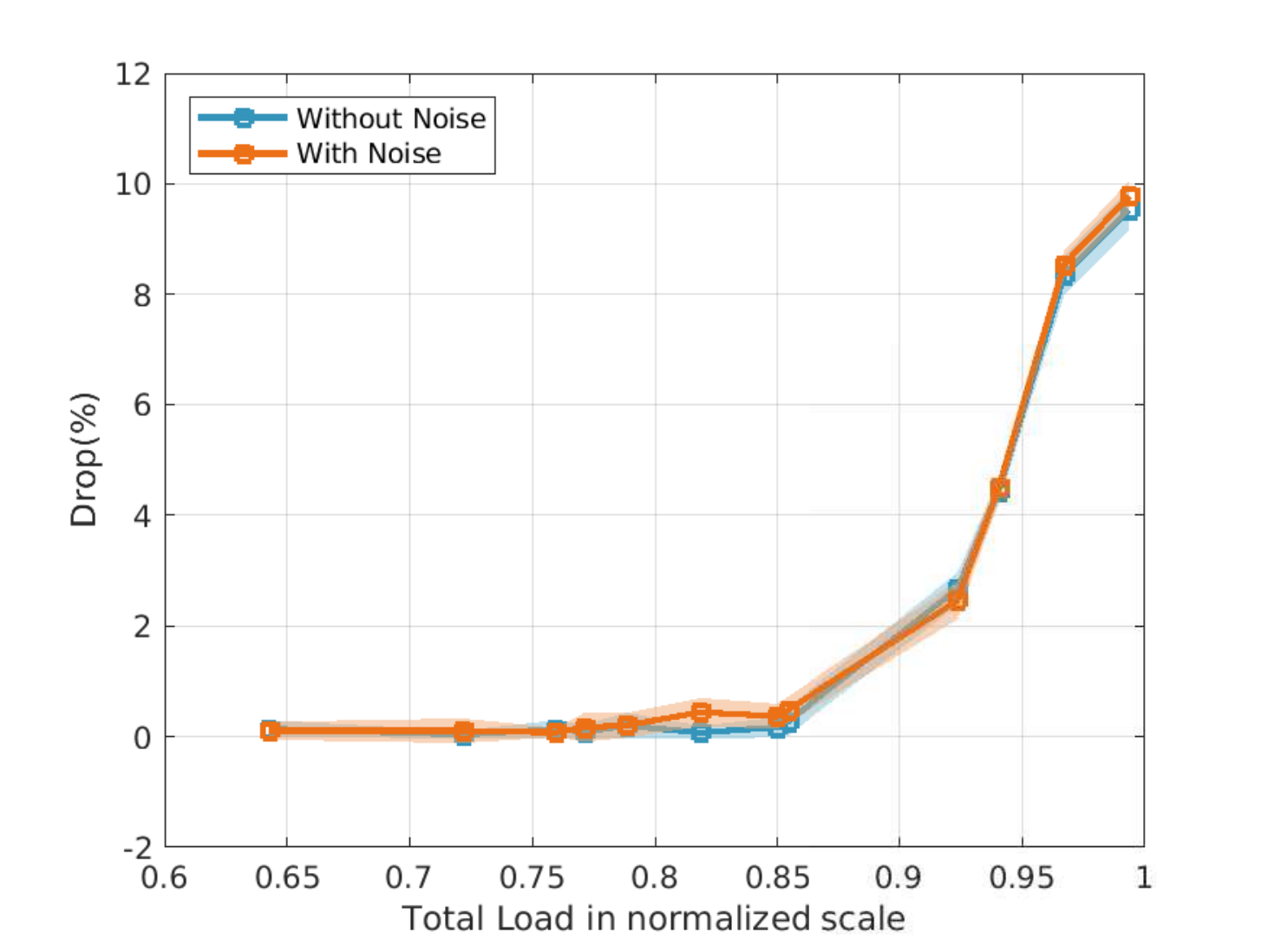}
		\caption{Drop percentage  for class 2 with and without noise. Noise is Gaussian with mean 0 and variance 25.}
		\label{robust2}
	\end{figure}
	\section{Conclusion}
	In this paper, we addressed the problem of resource allocation for emerging applications with bursty traffic patterns, requiring delay tolerance in order of milliseconds and high bandwidth requirements. We introduced a mechanism that aids in tracking delay and subsequently formed an MPC based model for obtaining optimal allocation policies. It is important that the optimal allocation policy is determined with low computational complexity. To this purpose, we showed that the MPC problem could be implemented by solving a linear program. As an implementation exercise, we chose PON as the networking framework where an edge/fog node is installed to provide real-time Internet services. Our MPC based allocations outperform the existing ones with respect to average delay, throughput and delay violation metrics. Further, the obtained allocations scale well when the number of subscribers are increased. The obtained allocations also show robust properties when subject to noisy predictions. Interesting extensions would be to look for further reductions in computational complexity. Also, it could be interesting to see the effect of using MPC over learned traffic with help of LSTMs which have proven to be highly efficient time-series predictors.
	\label{conc}
	\appendix
\section*{Proof of Total Unimodularity (Proposition 2)}
\begin{definition}
	A matrix is totally unimodular (TU) if the determinant of any of its square sub-matrices is either $0,1, or -1$.
\end{definition}
\textit{\textbf{Lemma 1}: If $\mathbf{A}$ is TU and $\mathbf{b}$ is integral, $\max\{\mathbf{cx}|\mathbf{Ax}\leq \mathbf{b}\}$ has integral optima.}\cite{schrijver1998theory,wiki:Unimodular_matrix}

Evidently, to use this result, we need some characterization of such matrices which would allow us to identify a TU matrix. The following results will be useful in our efforts to prove our proposition.

\textit{\textbf{Lemma 2}: If $\mathbf{A}$ is TU, the matrix obtained by removing or adding unit columns (column of an identity matrix) or unit rows (row of an identity matrix) is also TU.} \cite{schrijver1998theory}

\textit{\textbf{Lemma 3 (Ghouila-Houri)}: A matrix $\mathbf{A}$ is TU if and only if for each subset $R$ of rows of A, there exists two disjoint sets $R_1$ and $R_2$ such that $\sum\limits_{i\in R_1}a_{ij} - \sum\limits_{i\in R_2}a_{ij}\in \{0,-1,1\}$. In other words for every subset of rows of $\mathbf{A}$, we can find two disjoint sets such that the difference of the sum of rows would result in a vector with entries $0,1$ or $-1$.} \cite{schrijver1998theory,wiki:Unimodular_matrix}

We construct the matrix $\mathbf{A}$ from the constraints of our problem and use the aforementioned results to prove the proposition. We consider the single class case. The arguments can be extended to multi-class as well. The constraint set for class $c$ is given by the equations:
	\begin{subequations}
		\begin{gather}
		Q_{i-1}^c(t+1) = Q_{i}^c(t)-x_i^c(t) \; \forall c,i\; 1 < i \leq K^c \label{statecon1} \\
		Q_{K^c}(t+1) = A^c(t), \: \forall t,c 	\label{statecon2} \\
		x_i^c(t) \leq Q_i^c(t), \: \forall t,c \label{queuecon} \\
		\sum_{i = 1}^{K^c} x_i^c(t) \leq \Lambda, \forall \: t \label{bcon1} \\
		\sum_{t=0}^{H}\sum_{i = 1}^{K^c} x_i^c(t)\leq \Lambda^c, \: \forall \: c \label{bcon2}
		\end{gather}
	\end{subequations}
It is important to note that the bound $Q_i^c(t)$ is not a constant but involves optimization variables from previous time slots. Hence we should use \eqref{statecon1},\eqref{statecon2} recursively and obtain variable free bounds as follows:

	\begin{align}
		\nonumber
		&x_i^c(H)\leq Q_i^c(H) = Q_{i+1}^c(H-1) - x_{i+1}^c(H-1) \\ 
		&\implies x_i^c(H) + x_{i+1}^c(H-1) \leq Q_{i+1}^c(H-1) 
		\end{align}
		Continuing in this way, we have two cases corresponding to $H<K^c$ and $K^c\leq H$:
		
	    Case1: $H<K^c$
	    
		\begin{align}
		\nonumber
		\sum_{j = i}^{H+i} x^c_j(H-j+i) &\leq Q_{H+i}^c(0), \; H+i<K^c \\
		\nonumber
		\sum_{j = i}^{K^c} x^c_j(H-j+i) &\leq Q_{K^c}^c(H-K^c+i), \; H+i\geq K^c
		\end{align}
		
	    Case2: $K^c \leq H$ 
	    
	    \begin{align}
	    \nonumber
		\sum_{j = i}^{K^c} x^c_j(H-j+i) &\leq Q_{K^c}^c(H-K^c+i)  \\
		\nonumber
		&= A^c(H-K^c+i-1)
		\end{align}
It is easy to see that in either case, we either end up with a bound which is dependent on initial queue states (corresponding to slot 0) or the arrivals in each slot which govern the subsequent queue states.
Now we are ready to construct the matrix. For purpose of illustration, we will use the case of $K^c = 3$ and $H = 3$. However, our arguments will be general and can be extended for any $K^c$ and $H$. For sake of simplicity, we omit the superscript $c$.
\setcounter{MaxMatrixCols}{20}
\begin{align*}
\mathbf{A} = \begin{pmatrix}
1 & 0 & 0 & 0 & 0 & 0 & 0 & 0 & 0 & 0 & 0 & 0 \\
0 & 1 & 0 & 0 & 0 & 0 & 0 & 0 & 0 & 0 & 0 & 0\\
0 & 0 & 1 & 0 & 0 & 0 & 0 & 0 & 0 & 0 & 0 & 0\\
0 & 1 & 0 & 1 & 0 & 0 & 0 & 0 & 0 & 0 & 0 & 0\\
0 & 0 & 1 & 0 & 1 & 0 & 0 & 0 & 0 & 0 & 0 & 0\\
0 & 0 & 0 & 0 & 0 & 1 & 0 & 0 & 0 & 0 & 0 & 0\\
0 & 0 & 1 & 0 & 1 & 0 & 1 & 0 & 0 & 0 & 0 & 0\\
0 & 0 & 0 & 0 & 0 & 1 & 0 & 1 & 0 & 0 & 0 & 0\\
0 & 0 & 0 & 0 & 0 & 0 & 0 & 0 & 1 & 0 & 0 & 0\\
0 & 0 & 0 & 0 & 0 & 1 & 0 & 1 & 0 & 1 & 0 & 0\\
0 & 0 & 0 & 0 & 0 & 0 & 0 & 0 & 1 & 0 & 1 & 0\\
0 & 0 & 0 & 0 & 0 & 0 & 0 & 0 & 0 & 0 & 0 & 1\\
1 & 1 & 1 & 1 & 1 & 1 & 1 & 1 & 1 & 1 & 1 & 1\\
1 & 1 & 1 & 0 & 0 & 0 & 0 & 0 & 0 & 0 & 0 & 0\\
0 & 0 & 0 & 1 & 1 & 1 & 0 & 0 & 0 & 0 & 0 & 0\\
0 & 0 & 0 & 0 & 0 & 0 & 1 & 1 & 1 & 0 & 0 & 0\\
0 & 0 & 0 & 0 & 0 & 0 & 0 & 0 & 0 & 1 & 1 & 1\\
\end{pmatrix}
\end{align*}
The first 12 rows of the matrix correspond to the constraints \eqref{queuecon}. The $13^{th}$ row corresponds to constraint \eqref{bcon2} while the last 4 rows realize constraints \eqref{bcon2} for each slot/horizon. The columns are arranged in order of queues followed by slot index i.e., $\{x_1(0),x_2(0),x_3(0),x_1(1),\dots,x_3(3)\}$. The variables corresponding to initial time slot $0$ and the ${K^{c}}^{th}
$ queues do not require recursive expansion and hence correspondingly have unit rows. For proving $\mathbf{A}$ is TU, we need not consider these rows, since adding them do not affect the TU property of the matrix (due to Lemma 2). Hence, we are left with the following matrix.
\begin{align*}
\mathbf{A'} = \begin{pmatrix}
0 & 1 & 0 & 1 & 0 & 0 & 0 & 0 & 0 & 0 & 0 & 0\\
0 & 0 & 1 & 0 & 1 & 0 & 0 & 0 & 0 & 0 & 0 & 0\\
0 & 0 & 1 & 0 & 1 & 0 & 1 & 0 & 0 & 0 & 0 & 0\\
0 & 0 & 0 & 0 & 0 & 1 & 0 & 1 & 0 & 0 & 0 & 0\\
0 & 0 & 0 & 0 & 0 & 1 & 0 & 1 & 0 & 1 & 0 & 0\\
0 & 0 & 0 & 0 & 0 & 0 & 0 & 0 & 1 & 0 & 1 & 0\\
1 & 1 & 1 & 1 & 1 & 1 & 1 & 1 & 1 & 1 & 1 & 1\\
1 & 1 & 1 & 0 & 0 & 0 & 0 & 0 & 0 & 0 & 0 & 0\\
0 & 0 & 0 & 1 & 1 & 1 & 0 & 0 & 0 & 0 & 0 & 0\\
0 & 0 & 0 & 0 & 0 & 0 & 1 & 1 & 1 & 0 & 0 & 0\\
0 & 0 & 0 & 0 & 0 & 0 & 0 & 0 & 0 & 1 & 1 & 1\\
\end{pmatrix}
\end{align*}

We use Result 3 to prove $\mathbf{A'}$ is TU. To do so, we first inspect the matrix and present some observations. The rows corresponding to \eqref{queuecon} (first 6 rows) have the property that for any two rows, either the relative position of 1's and 0's do not match (commonly termed as orthogonal vectors) or if the positions match, the two rows differ by 1's occurring at regular intervals of 0's. For example: if we have the first two rows, they have 1's and 0's in alternate positions. If we have the second and third rows, the positions of  1's and 0's match but they differ by a 1 appearing after a 0 following the pattern 101. When $K^c>3$, the pattern introduces more number zeros between 1's. Further, the row with all 1's can be realized by adding the last four rows since, for each of the last four rows have their 1's placed at different positions (they are also mutually orthogonal).

Suppose we choose a subset of rows from the set of rows. The rows can either correspond to \eqref{queuecon} (first 6 rows), the all 1's (corresponding \eqref{bcon2}) or that corresponding to \eqref{bcon1} (last 4 rows). While considering the rows corresponding to \eqref{queuecon}, the idea is to construct $R_1$ and $R_2$ so that the sum of rows in $R_1$ have maximum number of 1's in the resulting vector. Clearly, this can be done by selecting all orthogonal rows and the rows with greater number of 1's when the relative positions match. The rows having same relative positions with lesser number of 1's compared to others are placed in $R_2$. Due to this construction, it is evident that the difference of sum of rows in $R_1$ and $R_2$ will only result in entries with either 0 or 1. Next, if we have the row with all 1's in our choice of subset, we place it in $R_2$ so that the difference will always result in entries with 0, 1 or -1. If our choice of subset have rows corresponding to \eqref{bcon1} but do not contain the row with all 1's, we place it in $R_2$ else, we place the row with all 1's in $R_2$ and the rows corresponding to \eqref{bcon1} in $R_1$. Since these rows are mutually orthogonal and their sum can only result in consecutive 1's, the mentioned assignment will always leave the difference of the rows between $R_1$ and $R_2$ with entries in  $\{0,1,-1\}$.

For example: From the set of rows of $\mathbf{A'}$ let $R = \{2,3,4,5,6,7,8\}$ be a subset of rows. By our mentioned method, we first place rows $\{3,5,6\}$ in $R_1$ and $\{2,4\}$ in $R_2$. Note that the difference of sum of rows will only lead to entries 0 or 1. Now we have the row of all 1's and row 8 which has 1's as the first three entries. We place the row 7 in $R_2$ and row 8 in $R_1$. Thus, $R_1 = \{3,5,6,8\}$, $R_2 = \{2,4,7\}$. The difference of the sum is $[0,0,0,-1,-1,-1,0,-1,0,0,0,-1]$.

For multi-class scenario, more number of variables are introduced in the matrix for each time slot. As an example we can consider the scenario for two classes, wherein the number of variables for each slot becomes $K^1+K^2$. In this case, the constraint \eqref{bcon2} has to be implemented for each class and hence instead of a single row of all 1's, we have a two rows one of which  follows the pattern $K^1$ 1's followed by $K^2$ 0's while the other has  $K^1$ 0's followed by $K^2$ 1's. We can apply all our previous arguments except for the case when one of these two rows is in the chosen set of rows along with one or more rows corresponding to \eqref{bcon1}. In this case, there is a possibility that the sum of difference between the rows gives 2. However, this can occur only for one particular class since we have only one row corresponding to \eqref{bcon2}. To circumvent this issue, we can first construct the two row sets for each class separately corresponding to constraints \eqref{queuecon}. Earlier it was ensured that the sum of difference of these two sets would have entries 0 or 1. The class which might encounter the aforementioned problem can be readjusted by simply swapping its corresponding row sets so that the sum of difference of those would have entries 0 or -1 instead of 0 or 1. Now, the previously mentioned assignment steps for rows corresponding to \eqref{bcon1} and \eqref{bcon2} can be followed to yield the sum of difference in $\{0,1,-1\}$.

	\bibliographystyle{IEEEtran}
	\bibliography{references}
\end{document}